\documentclass[11pt]{article}
\usepackage[utf8]{inputenc}

\usepackage{amsmath,amsthm,amsfonts}
\usepackage{xspace}
\usepackage[hidelinks]{hyperref}
\usepackage{cleveref}
\usepackage{graphicx, xcolor}
\usepackage{relsize}
\usepackage{enumitem}
\usepackage{fullpage}
\usepackage{amssymb}
\usepackage{dsfont}
\usepackage{thm-restate}
\usepackage{mathtools}
\usepackage{pgfplots}
\usepackage{subcaption}
\usepackage{wrapfig}
\usepackage[boxed]{algorithm2e}

\usepackage[style=alphabetic,backend=bibtex,maxbibnames=20,maxcitenames=6,giveninits=true,doi=false,url=false]{biblatex}
\newcommand*{\citet}[1]{\AtNextCite{\AtEachCitekey{\defcounter{maxnames}{2}}} \textcite{#1}}

\newcommand*{\citep}[1]{\cite{#1}}

\bibliography{main}

\newcommand{\lp}{\left}
\newcommand{\rp}{\right}

\newcommand{\A}{M}

\newcommand{\DE}{\mathrm{DE}}

\newcommand{\ZDER}[1]{Z_{\DE,R}}
\newcommand{\simiid}{\stackrel{\mathrm{i.i.d.}}{\sim}}

\newcommand{\calX}{\mathcal{X}}
\newcommand{\vx}{\mathbf{x}}
\newcommand{\vy}{\mathbf{y}}
\newcommand{\veta}{\mathbf{\eta}}

\renewcommand{\eqref}[1]{Eq.~(\ref{#1})}

\newtheorem{theorem}{Theorem}[section]

\newtheorem{definition}[theorem]{Definition}

\newtheorem{lemma}[theorem]{Lemma}

\newtheorem{claim}[theorem]{Claim}

\newtheorem{corollary}[theorem]{Corollary}

\newtheorem{question}{Question}

\newcommand{\E}{\mathbb{E}}

\title{A bounded-noise mechanism for differential privacy}

\begin{document}
\author{Yuval Dagan (MIT, EECS)\footnote{dagan@mit.edu}\qquad Gil Kur (MIT, EECS)\footnote{gilkur@mit.edu}}
\maketitle
\thispagestyle{empty}

\begin{abstract}
    We present an asymptotically optimal $(\epsilon,\delta)$ differentially private mechanism for answering multiple, adaptively asked, $\Delta$-sensitive queries, settling the conjecture of Steinke and Ullman [2020]. 
    Our algorithm has a significant advantage that it adds independent bounded noise to each query, thus providing an absolute error bound.
    Additionally, we apply our algorithm in adaptive data analysis, obtaining an improved guarantee for answering multiple queries regarding some underlying distribution using a finite sample. Numerical computations show that the bounded-noise mechanism outperforms the Gaussian mechanism in many standard settings.
\end{abstract}

\section{Introduction}

Differential privacy provides a framework to publish statistics of datasets that contain users' information, while preserving their privacy. Here, one assumes an underlying dataset $\vx = (x_1,\dots,x_n) \in \calX$ where $x_i$ contains private information of user $i$. An analyst, which does not have access to the dataset, requests some statistics of the data. These statistics are provided by the dataset holder, however, by analyzing them, the analyst should not learn significant information on any specific datapoint $x_i$. 
We follow the standard framework of \emph{$(\epsilon,\delta)$ differential privacy} \cite{dwork2006calibrating,dwork2006our} where the parameter $\epsilon$ quantifies the typical level of privacy and $\delta$ is (intuitively) the probability that the algorithm fails to preserve privacy (formally defined in Section~\ref{sec:prelim}).

Perhaps the most well-studied problem in differential privacy is \emph{answering multiple queries}. The adaptive version is described by an interactive game between the dataset holder and the analyst: in each iteration $i=1,\dots,k$, the analyst submits a query $q_i \colon \calX^n \to \mathbb{R}$. Then, the dataset holder should provide an approximate answer $a_i \approx q_i(\vx)$. Providing the exact answer may cause a leakage of private information and a common approach is to output $a_i = q_i(\vx)+\eta_i$, where $\eta_i$ is a random noise, whose outcome is unknown to the analyst. The goal is to keep the magnitude of noise as low as possible while preserving privacy.

Clearly, it is impossible to preserve privacy while answering arbitrary queries accurately. This happens when particular datapoints have significant influence over the outcome of the query: Here, an accurate answer to the query would necessarily leak information on these datapoints. To avoid this issue, it is common to assume that each datapoint can change the outcome by at most $\Delta$. Formally, we use the standard notion of $\Delta$-sensitive queries: For any user $i$ and any datasets $\vx$ and $\vx'$ that differ only on entry $i$, we assume that $|q(\vx)-q(\vx')|\le \Delta$. For example, $q$ can be the average of some bounded statistic $h \colon \calX \to [0,1]$, where $q(x_1,\dots,x_n) = \frac{1}{n} \sum_{i=1}^n h(x_i)$ and $\Delta = 1/n$. 

Despite being a central problem in differential privacy, it was unknown what is the least amount of noise that should be added. This question can be formalized as follows:
\begin{question}
	Fix parameters $\epsilon,\delta,k$ and $\Delta$. What is the minimal noise-level $\alpha$, such that there is an $(\epsilon,\delta)$ differentially private algorithm that answers $k$ $\Delta$-sensitive queries with error at most $\alpha$?
	(namely, $\forall i=1\cdots k, |a_i - q_i(\vx)| \le \alpha$)
\end{question}
One of the earliest algorithms, the \emph{Gaussian mechanism} \cite{dwork2006our}, consists of adding independent Gaussian noise of standard deviation $\sigma = O(R)$, where $R := \Delta \sqrt{k \log 1/\delta}/\epsilon$. Namely, $a_i = q_i(\vx)+\eta_i$ where $\eta_i \sim N(0,\sigma^2)$ are independent Gaussians. This yields a high-probability error of 
\[
\alpha 
= \max_{i=1,\dots,k} |q_i(\vx)-a_i|
= \max_{i=1,\dots,k} |\eta_i|
\le O(R \sqrt{\log k}),
\]
since the maximum of $k$ Gaussian random variables with standard deviation $\sigma$ is bounded by $O(\sigma \sqrt{\log k})$ with high probability. This was the best known algorithm until long after, when \citet{steinke2016between} showed how to obtain an improved bound $\alpha \le O(R \sqrt{\log\log k})$, by applying the same Gaussian mechanism and adding a smart algorithmic step that truncates the most-erroneous answers (via the \emph{sparse vectors} algorithm). In the same paper, they also showed a lower bound of $\alpha \ge \Omega(R)$, for any $\delta \ge e^{-k}$. Later, \citet{ganesh2020privately} showed how to obtain an improved bound of $O(R\sqrt{\log\log\log k})$ by replacing the Gaussian distribution with a generalized Gaussian and applying the same truncation technique of \cite{steinke2016between}. We further note that for all $\delta \le e^{-k}$, the optimal noise level is different, and equals $\Theta(k/\epsilon)$, using an algorithm that is in fact $(\epsilon,0)$ private \cite{hardt2010geometry}. Yet, it remained open whether one can match the lower bound of \cite{steinke2016between} and achieve a noise of $\alpha = O(R)$ for $\delta \ge e^{-o(k)}$. This was raised as an open problem by\citet{openProblem}.

One feature common to these algorithms is that they rely on adding \emph{unbounded noise}, and then, possibly, making a correction. Such an approach has multiple obvious disadvantages: (1) All the above-discussed algorithms fail to give a definite bound on the error that holds with probability $1$; (2) The correction step (i.e. the sparse vector technique), if used, complicates the algorithm and (3) The numerical constants associated with the noise may significantly degrade if one uses correction techniques. 

To guarantee a bounded noise, various prior works \cite{liu2018generalized,holohan2020bounded} suggested to truncate known noise distributions such as the Gaussian and Laplace. Yet, this yields suboptimal algorithms and it is possible that specifically tailored bounded-noise distributions would provide better results. This gives rise to the following question:
\begin{question}
What are the best mechanisms that rely on adding i.i.d. bounded noise? Can they provide the asymptotically optimal noise rate? Can they yield a reduced noise in practical settings?
\end{question}

\subsection{Main Results}
In this paper, we provide a positive answer to the above two questions:

\begin{theorem}\label{thm:main-informal} 
	Let $k,n \in \mathbb{N}$, $\epsilon \in (0,1]$, $\delta \in [e^{-k/\log^2(k)\log^4\log(k)},1/2]$ and $\Delta > 0$. There exists an algorithm for answering $k$ adaptive $\Delta$-sensitive queries that is $(\epsilon,\delta)$ deferentially private and further, its error is bounded as follows:
	\[
	\max_{i=1,\dots,k} |a_i - q_i(\vx)| \le 
	O(R) := O(\Delta\sqrt{k\log(1/\delta)}/\epsilon),
	\quad\text{with probability $1$.}
	\]
	Further, this is attained using an algorithm that adds i.i.d. noise of bounded magnitude to each answer.
\end{theorem}

In addition to providing an optimal error, this algorithm has additional benefits: 
\begin{itemize}
\item The bound on the maximal error, $\max_i |a_i - q_i(\vx)| = O(R)$, holds with probabilty $1$! This provides to the analyst definite bounds on the true answer $q_i(\vx)$, which is significantly more convenient in some settings. In comparison, the previous algorithms discussed above only guarantee a high probability error bound, which degrades at least as fast as $R\sqrt{\log 1/\beta}$ for confidence level $1-\beta$.
\item The algorithm is simple: the noise added to each query is drawn i.i.d. from some simple closed-form density. This is compared to the previous algorithms discussed above that relied on an additional algorithmic truncation step.
\item It yields better bounds than the Gaussian mechanism in many practical settings. This can be shown using an algorithm that computes upper bounds on the optimal noise level that is required to achieve $(\epsilon,\delta)$ privacy (see Section~\ref{sec:simulations}; code available online).
\end{itemize}

We recall that our algorithm achieves an optimal noise only for $\delta \ge e^{-k/\log^2k \log^4\log k}$.
A subsequent (which was essentially concurrent) independent work of \cite{ghazi2020avoiding} provided an optimal rate for all $\delta \ge e^{-k}$, thus closing the gap between the upper and lower bounds that was left open for $\delta$ slightly larger than $e^{-k}$. Their algorithm smartly consists of permuting the queries and applying multiple stages of the sparse vector algorithm to reduce the noise. The advantages of our algorithm include the three items discussed above and the fact that it can answer adaptively asked queries, which also makes it applicable to adaptive data analysis.

Lastly, we argue that it is impossible to achieve an optimal bound using an algorithm that adds i.i.d. bounded noise for all $\delta \ge e^{-k}$: (Proof of Section~\ref{sec:lb})

\begin{theorem}[informal]\label{thm:lb-informal}
There is no $(\epsilon,\delta)$ differentially private algorithm that adds bounded i.i.d. noise, that is asymptotically optimal in the regime $\delta \ge e^{-\omega(k/\log^2 k)}$, where $\omega()$ denotes a strict asymptotic inequality (we assume in the proof that the noise density is unimodal, yet we believe that this assumption is redundant).
\end{theorem}

\subsection{Application to adaptive data analysis.}
Adaptive data analysis concerns of answering multiple adaptively asked queries on some underlying distribution $P$ over a domain $\calX$, while having access only to a finite i.i.d. sample $x_1,\dots,x_n \sim P$ \cite{dwork2015preserving,hardt2014preventing}. 
This scenario is common in statistics and machine learning, where an adaptive procedure or an algorithm are used to infer or learn properties of the distribution.

The standard setting can be formulated as an interaction between the dataset holder, that has access to $n$ i.i.d. samples from $P$, and a statistical analyst whose goal is to infer properties of the distribution. In each iteration $t=1,\dots,k$ the algorithm submits a \emph{statistical query} $q_i \colon \calX \to [0,1]$, and the goal of the dataset holder is to send an answer $a_i$ that approximate the expectation $q_i(P) = \E_{x\sim P}[q_i(x)]$. The queries are asked adaptively, namely, $q_i$ can depend on the previous answers $a_1,\dots,a_{i-1}$. The goal is to answer all the queries with low error, ensuring that with probability at least $1-\beta$, $|a_i - q_i(P)| \le \alpha$ for all $i$.
%

The straightforward approach is to answer each query $q_i$ using the sample average, outputting $a_i = \frac{1}{n}\sum_{j=1}^n q_i(x_j)$. This gives a valid result with high probability for non adaptively-asked queries, namely, if $q_1,\dots,q_n$ are given a priori. However, in case that they are asked adaptively, it is possible, and even likely in some scenarios, that they fit or adjust to the specifically-drawn sample. In such cases, the sample-mean will not provide a valid approximation to the true expectation $q_i(P)$. A solution suggested by \cite{bassily2021algorithmic} is to use a differentially private algorithm to answer the queries. Intuitively, this prevents the queries $q_i$ from fitting to the data, since the previous answers $a_1\cdots a_{i-1}$ are differentially private with respect to it. In particular a \emph{transfer theorem} of \cite{bassily2021algorithmic,jung2020new} yields guarantees for adaptive data analysis given guarantees for the underlying differentially private algorithm. Applying these results on the known algorithms, one obtains the following guarantee: for all $\alpha,\beta \in (0,1/2)$, there is a sample size
\begin{equation}\label{eq:def-n-adaptive-data}
n(\alpha,\beta) = O\lp(\min\lp(
\frac{\sqrt{k \log k \log^2(1/\alpha\beta)}}{\alpha^2},\ \frac{\sqrt{k \log\log k \log^3(1/\alpha\beta)}}{\alpha^2}\rp)\rp),
\end{equation}
and an algorithm that receives $n(\alpha,\beta)$ samples from some arbitrary distribution $P$, and answers $k$ adaptively asked queries, with a high-probability error bound of
\[
\Pr\lp[\forall i=1\cdots k,\ 
|q_i(P)-a_i| \le \alpha
\rp]\ge 1-\beta.
\]
Here, the first argument in the right hand side of \eqref{eq:def-n-adaptive-data} corresponds to the Gaussian mechanism and the second to the algorithm of \cite{steinke2016between}. In this paper, we obtain the following improved bound: (Proof in Section~\ref{sec:ada-proof})
\begin{corollary} \label{cor:ada}
For every $k \in \mathbb{N}$ and $\alpha,\beta \in (0,1/2)$ such that $\alpha\beta \ge 4e^{-k/\log^2 k\log\log^4 k}$, there exists an algorithm for answering $k$ adaptive statistical queries $q_i \colon \calX \to [0,1]$, with a sample size of  
\[
n = O\lp(
\frac{\sqrt{k \log(1/\alpha\beta)}}{\alpha^2}\rp),
\]
that satisfies
$\Pr\lp[\forall i=1\cdots k,\ 
|q_i(P)-a_i| \le \alpha
\rp]\ge 1-\beta.$
More generally, this bound is also valid for answering $1$-sensitive queries.\footnote{When answering a $\Delta$-sensitive query $q \colon \calX^n\to \mathbb{R}$, the goal is to provide an approximation to the expected value of the query taken over a random dataset, $q(P^n) := \E_{\vx \sim P^n}[q(\vx)]$.}
\end{corollary}
This yields an optimal dependence both on $k$ and $\beta$ \cite{bassily2021algorithmic}, while removing logarithmic factors in $k,\alpha,\beta$. Further, this algorithm can answers approximately twice as many queries as the Gaussian mechanism in a standard setting (see Section~\ref{sec:simulations}).



\section{The abstract theorem}\label{sec:abstract-thm}

We present an abstract statement that provides guarantees for bounded-noise distributions, assuming that they satisfy some differential inequalities. We use the following notation for bounded-noise mechanisms:

\begin{definition}\label{def:general}
    Given a function $f \colon (-1,1)\to (0,\infty)$, consider the continuous distribution $\mu_f$ with density
    \[
    \mu_f(\eta) = \frac{\exp(-f(\eta))}{Z_{f}},\quad \text{where }  Z_{f} = \int_{-1}^1 \exp^{-f(\eta)} d\eta.
    \]
    Further, for any $R>0$ denote by $\mu_{f,R}$ the scaling of $\mu_f$ by $R$, namely $\eta\sim \mu_{f,R}$ is obtained from sampling $\eta' \sim \mu_f$ and setting $\eta = R\eta'$. Equivalent, $\mu_{f,R}$ has density
    \[
    \mu_{f,R}(\eta) = \frac{\exp(-f(\eta/R))}{Z_{f,R}},\quad \text{where }  Z_{f,R} = \int_{-R}^R \exp^{-f(\eta/R)} d\eta = R Z_{f}.
    \]
    Define by $M_{f,R}$ the mechanism that adds to each query a noise drawn independently from $\mu_{f,R}$.
\end{definition}

Next, we present our abstract theorem that shows that some noise mechanisms are optimal, if $f$ satisfies some desired properties. There are two essential properties: (1) $\mu_f$ decays to zero in the neighborhoods of $-1$ and $1$, or, equivalently, $f(\eta) \to \infty$ as $\eta \to \pm 1$; and (2) $\mu$ does not decay too fast, which amounts to requiring that $|f'(\eta)|$ is bounded in terms of $f(\eta)$. In particular, we would like that $I(|f'(\eta)|) \le f(\eta)$ for the function $I$ defined below. Any function $f$ that satisfies these assumptions (and a couple more technical assumptions), yields DP mechanisms with asymptotically optimal error, for any $\delta \ge \delta^*_k$. The threshold $\delta^*_k$ improves (i.e. decreases) as the upper bound on $|f'(\eta)|$ improves, or, equivalently, as $I$ increases.

Formally, let $I \colon [0,\infty)\to[0,\infty)$ be a continuous function that satisfies the following properties: 
\begin{itemize}
    \item  $I(t) \le t$ and $I(t) \ge c \sqrt{t}$ for any $t\ge C$, where $C,c>0$ can be any constants independent of $t$.
    \item  $I(t)$ is increasing in $t$ and $I(t)/t$ is decreasing in $t$.
\end{itemize}
For example, $I(t) = t^\alpha$ for some $\alpha \in [1/2,1]$ or $I(t) = t/\log^\alpha t$ for $\alpha \ge 0$.
Now, we state some requirements on the function $f$ that appears in the definition above:
\begin{enumerate}
    \item $f$ is symmetric, i.e. $f(-\eta) = f(\eta)$; 
    \item $f$ diverges:  $\lim_{\eta\to 1^-}f(\eta)=\lim_{\eta\to -1^+} f(\eta) = \infty$. 
    \item Bounded first derivative: $I(|f'(\eta)|) \le f(\eta)$.
    \item Bounded second derivative: $|f''(\eta)| \le C f(\eta)^2$, where  $C>0$ can be any constant independent of $\eta$.
\end{enumerate}
Lastly, we define $\delta^*_k$. For this purpose, define $t^*$ as the unique solution to $t=kI(t)/2t$. Notice that such a unique solution exists as $I(t)/t$ is continuous and decreasing in $t$. Then, we define $\delta^*_k = e^{-I(t^*)/C_f}$ where $C_f>0$ is a constant depending only on $f$.
This yields the following theorem: (Proof in Section~\ref{sec:proof-upper})
\begin{theorem}\label{thm:general-func}

Let $I(t)$ and $f$ satisfy the conditions above, let $k \in \mathbb{N}$ and $\delta^*_k$ is defined as above.
Let $\Delta>0$, $\epsilon \in (0,1]$, $\delta \in [\delta^*_k,1/2]$ and define $R = C_f \Delta \sqrt{k\log 1/\delta}/\epsilon$ for some constant $C_f$ depending only on $f$. Then, the mechanism $M_{f,R}$ is $(\epsilon,\delta)$-differentially private for answering $k$ adaptive $\Delta$-sensitive queries.
\end{theorem}

As a corollary, we obtain the following guarantees for specific functions $f$: (Proof in Section~\ref{sec:cor-to-abst})
\begin{corollary}\label{cor:to-abstract}
The following functions $f$ yield mechanisms $M_{f,R}$ with an optimal value of $R=\Theta(\Delta \sqrt{k\log(1/\delta))}/\epsilon)$, for any 
$\delta \ge \delta^*_k$:
\begin{itemize}
    \item The function $f(\eta) = 1/(1-\eta^2)^p$ with $\delta^*_k = \exp\lp(-C(p) k^{p/(p+2)}\rp)$, for any $p \ge 2$, where $C(p)$ depends only on $p$.
    \item The function $f(\eta) = \exp\lp(\exp\lp(1/(1-\eta^2)\rp)\rp)$ with $\delta^*_k = \exp(-C k/(\log^2 k \log^4\log k))$, for some $C>0$.
\end{itemize}
\end{corollary}

\section{Preliminaries}\label{sec:prelim}

\paragraph{Neighboring datasets and $\Delta$-sensitive queries.}
Given a domain $\calX$ and $n \in \mathbb{N}$, a \emph{dataset} is any element of $\calX^n$. Two datasets $\vx$ and $\vx'$ are called \emph{neighbors} if $\vx$ and $\vx'$ differ on exactly one entry. 
Given $\Delta > 0$, a \emph{$\Delta$-sensitive query} is any function $q \colon \calX^n \to \mathbb{R}$ such that for any two neighboring datasets $\vx$ and $\vx'$, it holds that $|q(\vx) - q(\vx')| \le \Delta$.

\paragraph{Interactive and non-interactive query-answering.}
The interactive setting can be viewed as an interactive game between two parties: (1) a dataset holder, which has access to some dataset $\vx\in \calX^n$, and (2) an analyst that has no information on $\vx$. In each iteration $t=1,\dots,k$, the analyst submits to the dataset-holder some query $q_i\colon \calX^n \to \mathbb{R}$, who replies with an answer $a_i$, that approximates the value $q_i(\vx)$. Notice that here $q_i$ can depend only on the previous answers $a_1,\dots,a_{i-1}$. In comparison, in the non-interactive setting the analyst submits all the queries ahead of time, and then the dataset holder answers them.

\paragraph{Differential privacy.}
Fix some mechanism $M$ for answering $k$ $\Delta$-sensitive queries and fix $\epsilon,\delta>0$. 
Let $A$ denote any query-asking strategy of the analyst that defines each query $q_i$ as a function of the previous answers $a_1,\dots,a_{i-1}$. We say that $M$ satisfies $(\epsilon,\delta)$-differential privacy if for any two neighboring datasets $\vx$ and $\vx'$, any analyst $A$ and any subset $U \subseteq \mathbb{R}^k$ of possible answers,
\[
\Pr[(a_1,\dots,a_k) \in U \mid \vx, A]
\le e^{\epsilon} \Pr[(a_1,\dots,a_k) \in U \mid \vx',A] + \delta.
\]
Intuitively, the distributions over the answers given any two neighboring datasets are similar.
\section{Proof Sketch}\label{sec:sketch}

We provide a proof sketch for Theorem~\ref{thm:general-func}, assuming that $\Delta = 1$. To simplify the presentation, we assume that the queries $q_1,\dots,q_n$ are fixed and non-adaptive.
The proof consists of two steps: first, we reduce the problem to showing a concentration inequality on a sum of independent variables, and secondly, we bound this sum.
\subsubsection*{Reducing to a concentration inequality}

In this section, our goal is to show that it suffices to prove \eqref{eq:sketch-toshow} ahead, which corresponds to bounding a weighted sum of $f'(\eta_1),\dots,f'(\eta_k)$, for randomly drawn $\eta_1,\dots,\eta_k \simiid \mu_f$.
Denote by $\vec q := (q_1,\dots,q_k)$ and $\vec a := (a_1,\dots,a_k)$ the vectors of queries and answers, respectively, and let $\Pr_{\cdot \mid \vx}$ and $\mathrm{density}_{\cdot \mid \vx}$ denote the conditional probability and density of $\vec a$ given the dataset $\vx$, respectively.
Recall that we want to show that for any two neighboring datasets $\vx$ and $\vy$ and any subset $U \subseteq \mathbb{R}^k$, we have that $\Pr_{\cdot \mid \vx}[\vec{a} \in U] \le e^\epsilon \Pr_{\cdot \mid \vy}[\vec{a}\in U]+\delta$.
A simple argument shows that it suffices to prove that
\begin{equation}\label{eq:first-comparison}
\Pr_{\cdot \mid \vx}\lp[\frac{\mathrm{density}_{\cdot \mid \vx}[\vec a]}{\mathrm{density}_{\cdot \mid \vy}[\vec a]} \ge e^{\epsilon}\rp] \le \delta
\enspace.
\end{equation}
Intuitively, this means that only a small fraction of the possible answers $\vec a \in \mathbb{R}^k$ are significantly more likely given $\vx$ compared to $\vy$. Recall that the answers $a_i$ are obtained by adding an i.i.d. noise whose density equals $\mathrm{noise}(\eta_i) := \exp(-f(\eta_i/R))/Z_{f,R}$, therefore 
\[
\mathrm{density}_{\cdot \mid \vx}[\vec a] 
= \prod_{i=1}^k \mathrm{noise}(a_i - q_i(\vx))
= \prod_{i=1}^k \exp\lp(-f\lp(\frac{a_i-q_i(\vx)}{R}\rp)\rp)/Z_{f,R}\enspace.
\]
Substituting this in \eqref{eq:first-comparison} and taking a log inside the $\Pr[]$, one obtains
\begin{equation}\label{eq:second-in-sketch}
\Pr_{\cdot \mid \vx}\lp[-\sum_{i=1}^k f\lp(\frac{a_i-q_i(\vx)}{R}\rp) + \sum_{i=1}^k f\lp(\frac{a_i-q_i(\vy)}{R}\rp) \ge \epsilon\rp] \le \delta.
\end{equation}
We substitute $\eta_i = (a_i - q_i(\vx))/R$ and $v_i = (q_i(\vx)- q_i(\vy))/R$, which also implies that $(a_i - q_i(\vy))/R = \eta_i + v_i$. Notice that $\eta_i \sim \mu_{f,1}=\mu_f$, namely $\eta_i$ is drawn from the normalized noise supported in $(-1,1)$ and notice that $v_i \in [-1/R,1/R]$, since $q_i$ is $1$-sensitive.
Then, \eqref{eq:second-in-sketch} translates to
\begin{equation}\label{eq:draft-sum-diff}
\Pr_{\vec \eta \sim \mu_{f}^k}\lp[\sum_{i=1}^k f\lp(\eta_i + v_i \rp) - \sum_{i=1}^k f\lp(\eta_i \rp) \ge \epsilon\rp] \le \delta.
\end{equation}
We then use the second-degree Taylor expansion to obtain 
$
f(\eta_i+v_i) = f(\eta_i) + v_i f'(\eta_i) + v_i^2 f''(\xi_2)/2
$
for some $\xi_i$ in the line connecting $\eta_i$ and $\eta_i+u_i$, and particularly, $\xi_i \in [\eta_i-1/R,\eta_i+1/R]$. Substituting this in \eqref{eq:draft-sum-diff} and substituting $v_i = u_i/R$, it suffices to prove the second inequality below:
\begin{equation}\label{eq:sketch-toshow}
\Pr_{\vec \eta \sim \mu_{f}^k}\lp[\sum_{i=1}^k v_i f'\lp(\eta_i\rp) + \sum_{i=1}^k v_i^2 f''(\xi_i) \ge \epsilon\rp]
\le
\Pr_{\vec \eta \sim \mu_{f}^k}\lp[\frac{1}{R}\sum_{i=1}^k u_i f'\lp(\eta_i\rp) + \frac{1}{2R^2}\sum_{i=1}^k \max_{\xi_i} |f''(\xi_i)| \ge \epsilon\rp] \le \delta
\end{equation}
where $u_i \in [-1,1]$ and the maximum is taken over 
$\xi_i \in [\eta_i-1/R,\eta_i+1/R]$. Notice that it is possible that $\xi_i \notin (-1,1)$, and for these values, we use the convention $f''(\xi_i) = \infty$.
We will bound separately by $\epsilon/2$ the sums that correspond to the first and the second derivatives. 

\subsubsection*{Proving the concentration inequality.}




Before sketching the actual concentration inequality that is used to bound the sum of first derivatives, we give an intuition by applying a central limit theorem, which is valid for any fixed $\delta$ as $k \to \infty$.

\paragraph{Central limit theorem for the sum of first-derivatives.}
Here, we assume for simplicity that $u_i \in \{-1,1\}$, which implies that $\mathrm{Var}(u_if'(\eta_i)) = \mathrm{Var}(f'(\eta_i)) :=\sigma^2$. Further,
notice that since $f$ is assumed to be symmetric, we have that $\E[u_i f'(\eta_i)] = u_i \E[f'(\eta_i)] = 0$ for all $i$.
Thus, for any $t \ge 0$,
\[
\lim_{k\to \infty} \Pr\lp[\frac{\sum_{i=1}^k u_i f'(\eta_i)}{\sqrt{k}\sigma} > t\rp] = \int_t^\infty \frac{e^{-s^2/2}}{\sqrt{2\pi}}ds \le e^{-t^2/2}~; \quad \text{where } \sigma^2 = \mathrm{Var}(f'(\eta_i))\enspace.
\]
If we fix $\delta > 0$, take $t = \sqrt{\log (2/\delta)}$ and $R = 2\sigma\sqrt{k\log(2/\delta)}/\epsilon$, we obtain that for a sufficiently large $k$,
\[
\Pr\lp[\frac{\sum_{i=1}^k u_i f'(\eta_i)}{R} > \epsilon/2 \rp] \le \delta/2.
\]
This is what we wanted to prove, in terms of the sum over first derivatives, and if we prove a similar statement with respect to the sum of second derivatives, the proof concludes. Yet, this bound holds for any \emph{fixed $\delta$} in the limit $k \to \infty$. Instead, we want a bound that holds when $k\to \infty$ and $\delta \to 0$ \emph{simultaneously}.


\paragraph{Non-asymptotic bound for the sum of first derivatives.}

Here, we would like to prove a non-asymptotic result. The standard approach to bounding a sum of independent random variables, $\sum_{i=1}^k X_i$, is to prove that each individual variable $X_i$ concentrates, and this should imply a concentration inequality for the sum. Perhaps the most well-known concentration inequality is Chernoff-Hoeffding, which assumes that the variables $X_i$ are bounded. Other inequalities assume that the $X_i$ have a bounded tail. For example, Bernstein's inequality is valid if there exists some constant $C>0$ such that
\begin{equation}\label{eq:subexp}
\Pr[|X_i|>t] \le C\exp(-t/C).
\end{equation}
In our case, substituting $X_i = u_i f'(\eta_i)$, we cannot guarantee such behavior. Instead, we can guarantee
\begin{equation}\label{eq:sketch-It-bnd}
\Pr[|X_i|>t] = \Pr[|u_i f'(\eta_i)|>t] \le C\exp(-I(t)),
\end{equation}
where $I(t)\ll t$ is the function given in the theorem statement. Next, we describe how to obtain \eqref{eq:sketch-It-bnd} and then, we explain how to bound the sum $\sum_i u_i f'(\eta_i)$ assuming \eqref{eq:sketch-It-bnd}. 

To prove \eqref{eq:sketch-It-bnd}, one can use the assumption $f(\eta_i) \ge I(|f'(\eta_i)|)$ and the fact that $I(t)$ is monotonic non-decreasing, and integrate:
\begin{align*}
\Pr[|u_i f'(\eta_i)|>t]
&\le \Pr[|f'(\eta_i)|>t]
= \frac{1}{Z_f} \int_{\substack{\eta \in (-1,1)\colon\\ |f'(\eta)| \ge t}} e^{-f(\eta)} d\eta
\le \frac{1}{Z_f} \int_{\substack{\eta \in (-1,1)\colon\\ |f'(\eta)| \ge t}} e^{-I(|f'(\eta)|)} d\eta \\
&\le \frac{1}{Z_f} \int_{\substack{\eta \in (-1,1)\colon\\ |f'(\eta)| \ge t}} e^{-I(t)} d\eta
\le \frac{2}{Z_f} e^{-I(t)}.
\end{align*}

Next, we explain how to bound the a sum of independent variables satisfying \eqref{eq:sketch-It-bnd}. Here we go along the lines of \cite{bakhshizadeh2020sharp} which uses the known idea of \emph{truncation}, as explained below. We start by explaining the standard approach that is used for bounded random variables or variables satisfying \eqref{eq:subexp}, and then explain how to adapt these ideas to our setting.
The standard approach is via an analysis of the moment generating function: for any $\theta > 0$, we can compute
\[
\E\lp[ \exp\lp(\theta \sum_{i=1}^k X_i\rp)\rp]
= \prod_{i=1}^k \E\lp[ \exp(\theta X_i)\rp],
\]
and use Markov's inequality to bound:
\[
\Pr\lp[\sum_i X_i > t \rp]
= \Pr\lp[\exp\lp(\theta\sum_i X_i\rp) > \exp(\theta t) \rp]
\le \frac{\E\lp[\exp\lp(\theta \sum_i X_i\rp) \rp]}{\exp(\theta t)}
= \frac{\prod_{i=1}^k\E\lp[\exp\lp(\theta X_i\rp) \rp]}{\exp(\theta t)}\enspace.
\]
We can now optimize over $\theta>0$ to obtain the known inequalities. 

Next, we move to our setting, substituting $X_i = u_i f'(\eta_i)$.
Since $I(t) \ll t$ and \eqref{eq:subexp} does not hold, we cannot use the MGF bound, because $\E[\exp(\theta X_i)] = \infty$ for all $\theta > 0$. This is called the \emph{heavy-tailed} regime. A standard approach is to truncate the random variables. Given some fixed $L \ge 0$, we define $X_i^{\le L} = X_i \mathds{1}(X_i \le L)$ and notice that $X_i^{\le L}$ is bounded, hence its MGF is finite for any $\theta>0$. Then, in order to bound the probability that $\sum_i X_i>t$, we first bound the probability that $\sum_i X_i^{\le L}>t$ and then bound the probability that there exists an $i$ such that $X_i > L$, leading to the following bound:
\[
\Pr\lp[ \sum_{i=1}^k X_i > t \rp]
\le \Pr\lp[ \sum_{i=1}^k X_i^{\le L} > t \rp]
+ \sum_{i=1}^k \Pr[X_i > L].
\]
The first term can be bounded using the moment generating function, and the second term is simply bounded by $Cke^{-I(L)}$, using \eqref{eq:sketch-It-bnd}. Optimizing over the parameter $\theta$ in the MGF bound and over the truncation parameter $L$, one obtains the following bound, assuming that the $X_i$ have zero mean:
\begin{equation}\label{eq:sketch-concentration}
\Pr\lp[ \sum_{i=1}^k X_i > t \rp]
\le 2e^{-t^2/kC'} + C'k e^{-I(t^*)/C'} \quad \text{ for all $t \le t^*$ where $t^*$ is the solution to } t = kI(t)/2t.
\end{equation}
Here, $t^*$ is the same parameter as defined in Theorem~\ref{thm:general-func} and $C'>0$ possibly depends on $I(t)$. We note that the first term is the analogue of the CLT, and it dominates the second term for $t \le t^*$, hence, we obtain the desired bound.

\paragraph{Bounding the sum of second derivatives.}

Recall that we want to show that
\begin{equation}\label{eq:toshow-second}
\Pr\lp[\frac{1}{2R^2}\sum_{i=1}^k \max_{\xi_i\in [\eta_i-1/R,\eta_i+1/R]} |f''(\xi_i)| \le \epsilon/2\rp] \ge 1-\delta/2.
\end{equation}
First, using the condition that $|f''(\eta)| \le f(\eta)^2$, it suffices to show that
\begin{equation}\label{eq:f-squared}
\Pr\lp[\frac{1}{2R^2}\sum_{i=1}^k \max_{\xi_i} f(\xi_i)^2 \le \epsilon/2\rp] \ge 1-\delta/2. 
\end{equation}

Next, we show that with high probability over $\eta_i$, we have that $f(\xi_i) \le 2 f(\eta_i)$. Since the density is proportional to $\exp(-f(\eta_i))$, $f(\eta_i)$ is small with high probability, so it is sufficient to show that if $f(\eta_i)$ is not too large, then $f(\eta_i \pm 1/R) \le 2 f(\eta_i)$. For that purpose, we use the condition that $|f'(\eta_i)| \le f(\eta_i)^2$, which guarantees that if $f$ is not very large, then it cannot grow very fast. This will conclude that with high probability $f(\xi_i) \le 2f(\eta_i)$, and by taking a union bound, this holds with high probability simultaneously for all $i$. Then, \eqref{eq:f-squared} translates to
\[
\Pr\lp[\frac{1}{2R^2}\sum_{i=1}^k (2f(\eta_i))^2 \le \epsilon/2\rp] \ge 1-\delta/2. 
\]
Similarly to the arguments above regarding the first derivative, we can show that $\Pr[f(\eta_i)^2>t] \le C\exp(-\sqrt{t})$. Again, we use a concentration inequality similar to \eqref{eq:sketch-concentration} to bound this sum. Following \eqref{eq:sketch-toshow} and the bound on the sum of first derivatives, this concludes the proof.

\section{Simulations}\label{sec:simulations}

We compare the numerical noise-levels of bounded-noise mechanisms to the Gaussian mechanism, for fixed values of $k,\epsilon$ and $\delta$. We used a computer program to derive tighter noise bounds than the ones appearing in the proof, by an exact computation of a suitable moment generating function (the formal derivation appears in Section~\ref{sec:exp-formal}; code appears online\footnote{Code available in \url{https://github.com/yuvaldag/Bounded-Noise-DP}}). 
We note that similar techniques can be used to obtain bounds on any mechanism that uses i.i.d. noise. Yet, for the Gaussian mechanism we used the exact optimal noise level, computed in \cite{balle2018improving}. Still, the upper bounds on the bounded-noise mechanism outperforms those optimally computed values for the Gaussian mechanism. We did not compare to the other mechanisms that have better asymptotic noise than the Gaussian mechanism, as the constants associated with their bounds are significantly worse. 

The comparison appears in Figure~\ref{fig:privacy}. For the bounded mechanism, we plot both the absolute bound on the noise and the $0.95$ probability bound on the maximal noise over $k$ queries, whereas for the Gaussian mechanism we plotted high probability bounds on the maximal noise with different confidence levels. It is worth noting that the gap between the bounded noise and the Gaussian mechanism increases as $k$ grows, as expected. For the fixed setting of $\epsilon=0.1$ and $\delta=10^{-10}$, the $0.95$-probability bound for the bounded-noise mechanism matches the $0.95$ bound by the Gaussian mechanism already at $k=10^3$, and it is $29\%$ less at $k=10^6$. Further, the absolute bound for the bounded-noise mechanism is lower by $28\%$ than the $0.999$-probability bound of the Gaussian mechanism at $k=10^6$.

Further, we present numerical comparisons for adaptive data analysis in Figure~\ref{fig:ada}. We used the same setting that was plotted in \cite{jung2020new}: we set the values of $\alpha = 0.1$ and $\beta = 0.05$, and for multiple values of $n$, we computed the number of number of adaptive queries that can be answered while keeping all the errors below $\alpha$ with probability $1-\beta$. Here, the bounded noise mechanism can answer at least twice many queries as the Gaussian mechanism for any $n \ge 8\cdot 10^5$ and it significantly outperforms the Gaussian mechanism also for smaller values of $k$.

\begin{figure}
\begin{subfigure}{0.33\textwidth}
\includegraphics[width=\linewidth]{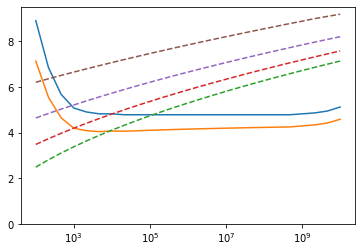}
\caption{$k$ ranges from $100$ to $10^{10}$\\ $\epsilon=0.1$ and $\delta=10^{-10}$}
\end{subfigure}
\begin{subfigure}{0.33\textwidth}
\includegraphics[width=\linewidth]{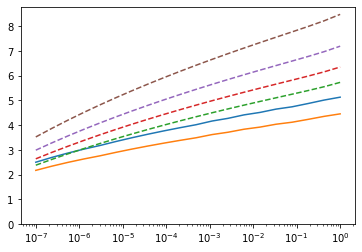}
\caption{$\epsilon$ ranges from $10^{-7}$ to $1$\\$k=10^6$ and $\delta=10^{-10}$}
\end{subfigure}
\begin{subfigure}{0.33\textwidth}
\includegraphics[width=\linewidth]{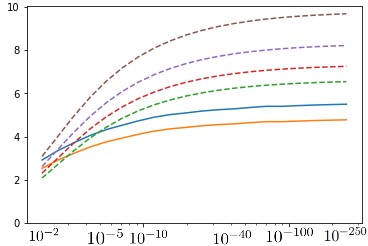}
\caption{$\delta$ ranges from $10^{-2}$ to $10^{-250}$\\$k=10^6$ and $\epsilon=0.1$}
\end{subfigure}
\caption{The errors of different mechanisms are plotted as a function of $k$, $\epsilon$ and $\delta$. The solid lines correspond to bounded noise mechanism with $f(\eta)=1/(1-\eta^2)^2$. In all of the plots, the upper solid line corresponds to the absolute bound on the noise, while the lower solid line corresponds to a $0.95$-probability bound on the maximal error over the $k$ queries. The dashed lines corresponds to the Gaussian mechanism, and they correspond to bounds on the maximal error that hold with probabilities $1-10^{-6},0.999.0.95$ and $0.5$ (larger noise corresponds to a higher probability). The values on the $x$-axis are described in each figure separately and $y$-axis corresponds to the noise divided by $\sqrt{k\log(1/\delta)}/\epsilon$.}
\label{fig:privacy}
\end{figure}

\begin{figure}[h!] 
    \centering
    \includegraphics[width=0.5\textwidth]{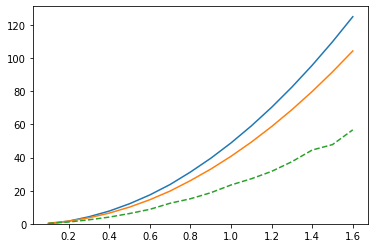}
    \caption{The number of queries $k$ (in thousands) that can be answered as a function of $n$ (in millions), while retaining $(\alpha=0.1,\beta=0.05)$-validity. The top line corresponds to the bounded noise mechanism with $f(\eta)=1/(1-\eta^2)$, the middle line to $f(\eta)=1/(1-\eta^2)^2$ (which is the same mechanism tested in Figure~\ref{fig:privacy}) and the dashed line corresponds to the Gaussian mechanism.}
    \label{fig:ada}
\end{figure}




\section{Abstract upper bound: Proof of Theorem~\ref{thm:general-func}}\label{sec:proof-upper}

First, notice that in the proof sketch we assume the queries to be non-adaptive. Hence, we start by explaining the differences that one has to make in order to adjust to the adaptive setting. Then, we proceed with the formal proof.
\paragraph{The adaptive vs. the non-adaptive setting.}
In the non-adaptive setting, the queries $q_1,\dots,q_n$ are asked ahead of time, and assumed to be fixed. Notice that in \eqref{eq:sketch-toshow}, the queries $q_i$ come into play via $u_i$. Indeed, recall that $u_i = R v_i= q_i(\vx) - q_i(\vy)$. In this setting, the $u_i$ are fixed numbers. Since the noise entries $\eta_i$ are i.i.d., we derive that $\sum_i u_i f'(\eta_i)$ is a sum of i.i.d. random variables. In order to bound them, we apply a concentration inequality for a sum of i.i.d. variables. 

In comparison, in the adaptive setting, $q_i$ is asked after observing the previous answers $a_1,\dots,a_{i-1}$. Since $q_i$ depends on $a_1,\dots,a_{i-1}$, then $u_i$ depends on $\eta_1,\dots,\eta_{i-1}$. In particular, the summands in $u_i f'(\eta_i)$ are no longer i.i.d. Yet, since $q_i$ is only a function of $\eta_1,\dots,\eta_{i-1}$, then $u_i$ is only a function of $\eta_{1},\dots, \eta_{i-1}$ and it is independent on $\eta_{i},\dots,\eta_k$. Since the $\eta_i$ variables are i.i.d., it holds that $\E[u_i f'(\eta_i) \mid \eta_1,\dots,\eta_{i-1}]=u_i \E[f'(\eta_i) \mid \eta_1,\dots,\eta_{i-1}]=0$. In particular, the partial sums of $\sum_i f'(\eta_i) u_i$ constitute of a Martingale whose deviation can be bounded using a concentration inequality.

\subsubsection*{Formal proof}
Here, we use asymptotic notation, e.g. $O()$, to hide constants that might depend on the log-density function $f$. Notice that it suffices to prove for $\Delta = 1$ (we can always scale the the queries and the noise by the same amount, while retaining $(\epsilon,\delta)$-privacy).

We start with a simple sufficient condition for the mechanism to be $(\epsilon,\delta)$ differentially private:
\begin{restatable}{lemma}{lemgeneral}
	\label{lem:general}
	Let $\epsilon,\delta > 0$ and let $k\in\mathbb{N}$. Let $f$ and $\vec \eta=(\eta_1,\dots,\eta_k) \overset{iid}{\sim} \mu_{f,1}$. Let $R>0$ and assume that for any random variables $v_1,\dots,v_n \in [-1/R,1/R]$ such that $v_i$ is a deterministic function of $\eta_1,\dots,\eta_{i-1}$, it holds that
	\begin{equation}\label{eq:1}
	\Pr\lp[\sum_{j=1}^k f(\eta_j + v_j) \le \sum_{j=1}^k f(\eta_j) + \epsilon\rp]
	\ge 1-\delta.
	\end{equation}
	Then, $M_{f,R}$ is $(\epsilon,\delta)$ differentially private.
   \end{restatable}
   \begin{proof}
	First, we can assume that $Z_{f,R} = \int_{-R}^R e^{-f(\eta/R)}d\eta= 1$ (otherwise, $f$ can be replaced with $f + \log Z_{f,R}$).
	Let $\vx$ and $\vx'$ denote two neighboring datasets, let $\vec a=(a_1,\dots,a_k)$ denote the random output of the algorithm on input $\vx = (x_1,\dots,x_n)$ denote by $\vec a'=(a_1',\dots,a_k')$ its output on input $\vx'$.
	Let $U \subseteq \mathbb{R}^k$ and our goal is to show that $\Pr[\vec a \in U] \le e^\epsilon \Pr[\vec a'\in U]+\delta$. 
	Denote $v_j = (q_j(\vx) - q_j(\vx'))/R$ for $j \in [k]$ and notice that $v_j$ is a deterministic function of $\eta_1,\dots,\eta_{j-1}$. Denote 
	\[
	G = \lp\{ \vec{\eta} \in \mathbb{R}^k \colon 
	\sum_j f(\eta_j) \ge \sum_j f(\eta_j + v_j) - \epsilon \rp\}\enspace.
	\]
	Notice that by \eqref{eq:1}, $\Pr[\vec{\eta}\notin G] \le \delta$. Denote $(U-\vec{b})/c = \{(\vec{u}-\vec{b})/c \colon u \in U \}$ for any $\vec{b} \in \mathbb{R}^k$ and $c \ne 0$, define $\vec q(\vx) = (q_1(\vx),\dots,q_k(\vx))$, notice that $a_j = q_j(\vx) + R\eta_j$ and estimate:
	\begin{align*}
	\Pr[\vec a \in U]
	&= \Pr\lp[R\vec\eta \in \lp(U-\vec q(\vx)\rp)\rp]
	\le \Pr\lp[\vec{\eta} \in \lp(\lp(U-\vec q(\vx)\rp)/R\rp)\cap G\rp] + \Pr\lp[\vec{\eta}\notin G \rp]\\
	&\le \Pr\lp[\vec{\eta} \in \lp(\lp(U-\vec q(\vx)\rp)/R\rp)\cap G\rp] + \delta
	= \int_{\vec{u} \in \lp(\lp(U-\vec q(\vx)\rp)/R\rp)\cap G} e^{-\sum_i f(u_i)} du+\delta\\
	&\le \int_{\vec{u} \in \lp(\lp(U-\vec q(\vx)\rp)/R\rp)\cap G} e^{\epsilon -\sum_i f(u_i + v_i)} du+\delta
	= e^\epsilon \Pr\lp[\vec{\eta} \in \lp(\lp(U-\vec q(\vx) + R\vec{v}\rp)/R\rp)\cap G\rp] + \delta\\
	&= e^\epsilon \Pr\lp[\vec{\eta} \in \lp(\lp(U-\vec q(\vx')\rp)/R\rp)\cap G\rp] + \delta 
	\le e^\epsilon \Pr\lp[R\vec{\eta} \in \lp(U-\vec q(\vx')\rp)\rp] + \delta\\
	&= e^\epsilon \Pr[\vx'\in U] + \delta.
	\end{align*}
\end{proof}

To prove that \eqref{eq:1} holds, one can approximate $f$ by its second-degree Taylor expansion, thus deriving the following statement.
\begin{restatable}{lemma}{lemTaylor}\label{lem:Taylor}
	Let $\epsilon,\delta,R > 0$ and let $k\in\mathbb{N}$. Let $f$ and let $\vec \eta \simiid \mu_{f,1}$. Assume that for any $u_1,\dots,u_n \in [-1,1]$ such that $u_j$ is a deterministic function$\eta_{1},\cdots, \eta_{j-1}$, the following holds:
	\begin{equation}\label{eq:taylor}
	\Pr\lp[ 
	\forall i,\ \eta_i \in (-1+1/R, 1-1/R), \ \text{and }
	\lp|\sum_{i=1}^k \frac{u_i f'(\eta_i)}{R}\rp|
	+ \sum_{i=1}^k \max_{\xi_i \colon |\xi_i-\eta_i| \le 1/R} \frac{|f''(\xi_i)|}{2R^2}
    \le \epsilon
	\rp] \ge 1-\delta.
	\end{equation}
	Then, $M_{f,R}$ is $(\epsilon,\delta)$ differentially private.
\end{restatable}
\begin{proof}
	We will show that \eqref{eq:taylor} implies \eqref{eq:1}, where $u_i$ in \eqref{eq:taylor} is replaced with $Rv_i$ in \eqref{eq:1}. In particular, notice that we can write
	$f(\eta_j+v_j)$ using the Taylor expansion $f(\eta_j + v_j) = f(\eta_j) + f'(\eta_j)v_j + f''(\xi_j) v_j^2/2$ where $\xi_j$ is a point in the line connecting $\eta_j$ and $\eta_j+v_j$. We have
	\begin{align*}
	\sum_j f(\eta_j + v_j)
	= \sum_j f(\eta_j) + \sum_j f'(\eta_j) v_j + \sum_j \frac{f''(\xi_j) v_j^2}{2}\\
	\le \sum_j f(\eta_j) 
	+ \lp|\sum_j \frac{f'(\eta_j) u_j}{R}\rp| 
	+ \sum_j \max_{\xi_j \colon |\xi_j - \eta_j| \le 1/R}\frac{|f''(\xi_j)|}{2R^2}.
	\end{align*}
	Thus, whenever the high-probability event in \eqref{eq:taylor} holds, the event of \eqref{eq:1} also holds. In particular, \eqref{eq:taylor} implies \eqref{eq:1}, which concludes the proof.
\end{proof}

To apply Lemma~\ref{lem:Taylor}, we would like to prove concentration of the sum of the derivatives of $f$. In order to analyze the concentration properties of a sum of random variables, it is common to consider each variable separately and then use a concentration result for sums. We start by providing a definition of what it means for a random variable to concentrate:
\begin{definition}
Given $C>0$ and a function $I(t) \colon [0,\infty)\to\mathbb{R}$, we say that a random variable $X$ is $(I(t),C)$ bounded if for all $t$, $\Pr[|X|>t] \le C \exp(-I(t))$.
\end{definition}

Given a sum of $(I(t),C)$ bounded variables, we can obtain the following concentration inequality, that is proven in Section~\ref{subsec:probabilistic}, using ideas from \cite{bakhshizadeh2020sharp}.

\begin{restatable}{proposition}{propConcentrateMartingale}\label{prop:heavytail}
Let $X_1 \cdots X_n$ be a Martingale, namely, for all $i$, $\E[X_i \mid X_1,\dots,X_{i-1}] = 0$. Further, assume that there exists a function $I \colon [0,\infty) \to [0,\infty)$ and $C>0$ such that
\[
\Pr[|X_{i}| > t \mid X_1 \cdots X_{i-1}] \le C\exp(-I(t)),
\]
and additionally, $I(t) \le t$ and $I(t)/t$ is monotonic decreasing. Let $M>0$ be such that $M \ge C\int_0^\infty (t^2 + 2t) e^{-I(t)/2}$ and let $t^*$ be the unique solution to 
\[
t = MnI(t)/2t.
\]
Then, for any $t > 0$,
\[
\Pr\lp[\lp|\sum_{i=1}^n X_i\rp| > t\rp]
\le \begin{cases}
2 e^{-t^2/2Mn} + Cn e^{-I(t^*)} & t \le t^* \\
2 e^{-I(t)/4} + Cne^{-I(t)} & t > t^*
\end{cases}\enspace.
\]
\end{restatable}

To give some intuition, we note that for any $t \le t^*$, the first term dominates the second, and here, we get a sub-Gaussian concentration, namely, the tail behaves as a Gaussian tail, where the variance of the Gaussian is replaced with $Mn$, where $M$ is just a constant that depends on $I(t)$. Yet, one could not hope for a sub-Gaussian concentration for all $t\ge 0$. After all, a single variable $X_i$ decays slower than a Gaussian. At some point, the heavy tail of the single $X_i$ will dominate the sub-Gaussian tail of the sum, and this happens exactly at $t^*$. From that point onward, the tail is dictated by the function $I(t)$.

In order to apply Proposition~\ref{prop:heavytail}, we would like to show concentration properties of a single instance of $f'(\eta)$ and $f''(\eta)$. This follows from the fact that $f'$ and $f''$ are bounded in terms of some function of $f$, using the following simple lemma:
\begin{lemma}\label{lem:bnd-single}
Let $\eta$ be a random variable supported on $(-1,1)$, with density $e^{-f(\eta)}/Z$ where $Z$ is the normalizing constant and $f(\eta) > 0$. Assume that $h\colon (-1,1)\to\mathbb{R}$ is such that $f(\eta) \ge I(|h(\eta)|)$ for all $\eta$, for some increasing $I \colon (0,\infty) \to (0,\infty)$. Then, for any $t \ge 0$,
\[
\Pr[|h(\eta)| \ge t] \le \frac{2}{Z} e^{-I(t)}.
\]
\end{lemma}
\begin{proof}
A simple calculation shows
\[
\Pr[|h(\eta)| \ge t]
= \frac{1}{Z} \int_{\substack{\eta \in (-1,1)\colon\\ |h(\eta)| \ge t}} e^{-f(\eta)} d\eta
\le \frac{1}{Z} \int_{\substack{\eta \in (-1,1)\colon\\ |h(\eta)| \ge t}} e^{-I(|h(\eta)|)} d\eta
\le \frac{1}{Z} \int_{\substack{\eta \in (-1,1)\colon\\ |h(\eta)| \ge t}} e^{-I(t)} d\eta
\le \frac{2}{Z} e^{-I(t)}.
\]
\end{proof}

We can apply Lemma~\ref{lem:bnd-single} to prove a bound on the weighted sum of derivatives:
\begin{lemma} \label{lem:bnd-sum-derivatives}
    Let $t^*$ be as in the definition of Theorem~\ref{thm:general-func}. It holds that for any $t \le t^*$ and any sufficiently large $k$,
    \[
    \Pr\lp[ \lp|\sum_j f'(\eta_j) u_j\rp| \ge t \rp] \le e^{-t^2/2C_f k} + e^{-I(t^*)/2},
    \]
    where $C_f$ depends only on $f$.
\end{lemma}
\begin{proof}
    Applying Lemma~\ref{lem:bnd-single}, we have that $f'(\eta_i)$ is zero mean and is $(I(t),2/Z_{f,1})$ bounded, where $2/Z_{f,1}$ is a constant. Since $u_i$ and $\eta_i$ are independent conditioned on $\eta_1,\dots,\eta_{i-1}$ and since $|u_i|\le 1$, we derive that conditioned on $\eta_1,\dots,\eta_{i-1}$, the random variable $u_i f'(\eta_i)$ is also zero mean and $(I(t),2/Z_{f,1})$-bounded. We would like to apply Proposition~\ref{prop:heavytail}. Let us discuss what values we substitute in that proposition: 
    \begin{itemize}
    \item
    First, consider the value $M$, that has to be lower bounded by $C\int_0^\infty (t^2+2t)e^{-I(t)/2}dt$. 
    Notice that this integral converges due to the assumption in Theorem~\ref{thm:general-func} that $I(t) \ge \Omega(\sqrt{t})$. We will substitute $M$ to the maximum of that integral and $1$
    \item
    Next, notice that we substitute $n$ with $k$. 
    \item 
    Further, let us distinguish between the value $t^*$ appearing in Proposition~\ref{prop:heavytail}, that we will denote here by $t'$, which is the solution of
    $t = MkI(t)/2t$, and the value appearing Theorem~\ref{thm:general-func}, which is the solution of $t = kI(t)/2t$, that we denote here by $t^*$. We note that $t' \ge t^*$, since $M\ge 1$ and since $I(t)/t$ is monotonic decreasing in $t$.
    \end{itemize}
    We obtain that for any $t \le t^*$,
    \[
    \Pr\lp[\lp| \sum_{i=1}^k v_i f'(\eta_i) \rp|>t\rp]
    \le 2e^{-t^2/2Mk} + Cke^{-I(t^*)}.
    \]
    Let us bound the second term, $Cke^{-I(t^*)}$, and for that purpose, let us obtain a lower bound on $t^*$: first, for a sufficiently large $k$, it holds that $t^* \ge 1$. Indeed, for any $t \le 1$ and for a sufficiently large $k$, the value at the right hand side of the equation $t = kI(t)/2t$ is $kI(t)/2t \ge kI(1)/(2\cdot 1) \ge 1 \ge t$, which follows from the monotonicity assumption on $I(t)/t$. By definition of $t^*$ we have $t^* \ge 1$. From the definition of $t^*$ we have that  $t^* = kI(t^*)/2t^*$, hence, $t^* = \sqrt{kI(t^*)/2}$. For any sufficiently large $k$, $t^* \ge 1$, hence, since $I(t)$ is increasing, we have that $t^* \ge \sqrt{kI(1)/2}$.
    Further, recall that $I(t) \ge \Omega(\sqrt{t})$, which implies that $I(t^*)\ge \Omega(\sqrt{k})$, hence, for a sufficiently large $k$, 
    \[
    Cke^{-I(t^*)} = Cke^{-I(t^*/2)}\cdot e^{-I(t^*)/2} \le
    Ck e^{-\Omega(\sqrt{k})} \cdot e^{-I(t^*)/2}
    \le e^{-I(t^*)/2}.\]
\end{proof}

Next, we would like to bound the term that corresponds to the second derivative. Recall that our goal is to bound $\sum_i |f''(\xi_i)|$ where $\xi_i$ is in the vicinity of $\eta_i$. We start by bounding the sum $\sum_i |f''(\eta_i)|$ and then relate the sum over $\xi_i$ to that over $\eta_i$. Since $|f''(\eta)|\le O(f(\eta)^2)$ we can instead use the following lemma:
\begin{lemma}\label{lem:bnd-second-der}
    Let $\vec\eta \sim \mu_{f,1}^k$. Then, for any sufficiently large $k$ and any $t \ge k$,
    \[
    \Pr\lp[\lp| \sum_{i=1}^k f(\eta_i)^2 \rp| > t + Ck \rp] \le e^{-c\sqrt{t}},
    \]
    for some constants $C,c>0$ depending only on $f$.
\end{lemma}
\begin{proof}
    We would like to apply Proposition~\ref{prop:heavytail}, with the following substitutions:
    \begin{itemize}
        \item We replace $n$ with $k$ and $X_i$ with $f(\eta_i)^2 - \E f(\eta_i)^2$. Notice that $X_i$ is zero mean and $X_1,\dots,X_k$ are independent, hence $\E[X_i \mid X_1\dots X_{i-1}] = 0$ as required.
        \item From Lemma~\ref{lem:bnd-single}, we have that $X_i$ is $(\sqrt{t},2/Z_{f,1})$-bounded. In particular, we replace $I(t)$ with $\sqrt{t}$ and $C$ with $2/Z_{f,1}$.
        \item We replace $M$ with the corresponding integral in Proposition~\ref{prop:heavytail}, and this integral converges as argued in Lemma~\ref{lem:bnd-sum-derivatives}.
        \item We replace $t^*$ with with the solution of $t = Mk\sqrt{t}/2t$, and notice that $t* = (Mk/2)^{2/3}$.
    \end{itemize}
	Since $t^*= \Theta(k^{2/3})$ it follows that for a sufficiently large $k$, $k\ge t^*$. From Proposition~\ref{prop:heavytail} we obtain that for any $t \ge k$,
	\[
	\Pr\lp[ \sum_{i=1}^k f(\eta_i)^2 > t + \sum_{i=1}^k \E f(\eta_i)^2 \rp] \le 2e^{-c\sqrt{t}},
	\]
	for some constant $c>0$.
    Lastly, notice that from Lemma~\ref{lem:bnd-single}, we have that $\E[f(\eta_i)^2] < \infty$, hence, $\sum_{i=1}^k \E f(\eta_i)^2 \le O(k)$. This concludes the proof.
\end{proof}
This lets us bound $\sum_i |f(\eta_i)''|$, however, recall that we want a bound on $\sum_i |f''(\xi_i)|$ for some $\xi_i$ in the vicinity of $\eta_i$. In fact, it suffices to bound $f(\xi_i)$ in terms of $f(\eta_i)$ and then bound $f''(\xi_i) \le O(f(\xi_i)^2)$. Therefore, we have the following lemma:

\begin{lemma}\label{lem:differential-eq}
Let $f \colon (-1,1) \to (0,\infty)$ be a function such that $\lim_{ \eta \to 1^-} f(\eta) = \lim_{\eta \to -1^+} = \infty$, and $|f'(\eta)| \le C f(\eta)^2$ for some $C>0$. Then, for any $\eta \in (-1,1)$ and any $\lambda \in [-1,1]$,
\[
f(\eta) \ge f\lp(\eta+\frac{\lambda}{2Cf(\eta)}\rp)/2.
\]
\end{lemma}
\begin{proof}
First, assume that $\lambda \ge 0$. Fix some $\eta$, let $C$ be the constant such that $f'(\eta) \le C f(\eta)^2$ and define the function 
\[
g(\theta)
= \frac{1}{1/f(\eta) + C(\eta-\theta)}.
\]
Computing the derivative of $g$ with respect to $\theta$, one obtains
\[
g'(\theta) = \frac{C}{\lp(1/f(\eta) + C(\eta-\theta)\rp)^2}
= Cg(\theta)^2,
\]
for all $\theta$ such that $1/f(\eta) + C(\eta-\theta) > 0$. In particular, this holds for all $\theta < \eta + 1/(Cf(\eta))$. Notice that $f(\eta) = g(\eta)$ and further, that the assumption that $f'(\eta) \le Cf(\eta)^2$ while $g'(\theta) = Cg(\theta)^2$, implies that $f(\theta) \le g(\theta)$ for all $\theta \in [\eta, \eta + 1/(Cf(\eta)))$.
In particular,
\[
f\lp(\eta + \frac{\lambda}{2Cf(\eta)}\rp)
\le g\lp(\eta + \frac{\lambda}{2Cf(\eta)}\rp)
= \frac{f(\eta)}{1-\lambda/2}
\le 2f(\eta),
\]
as $\lambda \le 1$.

For the case that $\lambda < 0$, the result follows by applying the same lemma with $\tilde{\lambda} = -\lambda$, $\tilde{\eta} = -\eta$ and $\tilde{f}(\theta) = f(-\theta)$.
\end{proof}

We derive the following bound on the the sum of second derivatives as a consequence:
\begin{lemma} \label{lem:second-der-final}
    For a sufficiently large $k$, and any $t \ge k$, it holds with probability at least $1-e^{-\sqrt{t}/C_f}-e^{-R/C_f}$ that 
    \[
    \sum_{i=1}^k \sup_{\xi_i \colon |\xi_i-\eta_i|\le 1/R} |f''(\xi_i)|
    \le C_f t,
    \]
    where $C_f>0$ is a constant depending only on $f$.
\end{lemma}
\begin{proof}
    First, we bound $f''(\xi_i) \le Cf(\xi_i)^2$, as given in the assumptions of Theorem~\ref{thm:general-func}.
    Next, we would like to bound $f(\xi_i)$ in terms of $f(\eta_i)$. Notice that $|\xi_i - \eta_i| \le 1/R$. From Lemma~\ref{lem:differential-eq}, we have that if $1/2Cf(\eta) \ge 1/R$, then $f(\xi_i) \le 2 f(\eta_i)$. This happens whenever $f(\eta) \le R/2C$. 
    From Lemma~\ref{lem:bnd-single}, this happens with probability at least $1- 2/Z_{f,1}\cdot e^{-R/2C}$. By a union bound over the $k$ coordinates, and since $R \ge \sqrt{k}$, we have that if $k$ is sufficiently large, then with probability $1 - e^{-R/4C}$, all the $k$ coordinates satisfy $f(\eta) \le R/2C$. This implies that
    \[
    \sum_i |f''(\xi_i)| \le C\sum_i f(\xi_i)^2 \le 2C \sum_i f(\eta_i)^2.
    \]
    From Lemma~\ref{lem:bnd-second-der}, we have that w.p. $e^{-c\sqrt{t}}$, $\sum_i f(\eta_i)^2 \le t + O(k)$. Combining the above arguments, we obtain that with probability at least $1-e^{-c\sqrt{t}}-e^{-R/4C}$,
    \[
    \sum_i |f''(\xi_i)| 
    \le 2C \sum_i f(\eta_i)^2 \le 2C t + O(k) \le O(t),
    \]
    since we assumed in this lemma that $t \ge k$. This concludes the proof.
\end{proof}

\begin{proof}[Proof of Theorem~\ref{thm:general-func}]
From Lemma~\ref{lem:Taylor} it suffices to show that with probability $1-\delta$,
\begin{equation}\label{eq:final-bound}
\lp|\sum_{i=1}^k \frac{u_i f'(\eta_i)}{R}\rp|
	+ \sum_{i=1}^k \max_{\xi_i \colon |\xi_i-\eta_i| \le 1/R} \frac{|f''(\xi_i)|}{2R^2}
    \le \epsilon.
\end{equation}
The first term can be bounded by $\epsilon/2$, if we substitute $t=R\epsilon/2$ in Lemma~\ref{lem:bnd-sum-derivatives}, and the failure probability is bounded by
\[
e^{-R^2\epsilon^2/kC_f} + e^{-I(t^*)/2}.
\]
Since we assumed that $R \ge \Omega(\sqrt{k\log 1/\delta}/\epsilon)$, if the constant in the definition of $R$ is sufficiently large, then the first term can be bounded by $\delta/4$. For the second term, recall that we assume that $\delta \ge e^{-\Omega(I(t^*))}$. If the constant in the $\Omega()$ is sufficiently small, then this term is also bounded by $\delta /4$. We conclude that the weighted sum of derivatives is bounded by $\epsilon/2$ with probability at least $1-\delta/2$.

Next, we consider the The second term, that corresponds to the second derivatives. To bound it, we apply Lemma~\ref{lem:second-der-final}, substituting $t = C_0 k \log (1/\delta)$, where $C_0>0$ is a sufficiently large constant to be determined later. We derive that with probability $1-e^{-\sqrt{t/C_f}} - e^{-R/C_f}$,
\[
\sum_i |f''(\xi_i)| \le O(k \log (1/\delta)).
\]
Recall that in \eqref{eq:final-bound} this sum is divided by $R^2$. In particular, if the constant in the definition of $R$ is sufficiently large, then this term is bounded by $\epsilon/2$ as required. Lastly, notice that the failure probability is bounded by
\[
e^{-C_0 \sqrt{k\log(1/\delta)}/C_f} + e^{-C_0 R/C_f} \le 2e^{-C_0  \sqrt{k \log (1/\delta)}/C_f},
\]
assuming that the constant in the definition of $R$ is sufficiently large.
First, we would like to bound $k \ge \log 1/\delta$. Recall that $\delta \ge e^{-I(t^*)} \ge e^{-t^*}$, since $I(t)\le t$. From the inequality $I(t)\le t$ and the definition of $t^*$, we have that $t^* = kI(t^*)/2t^* \le k/2\le k$, which, in combination with $\delta \ge e^{-t^*}$, implies that $\delta \ge e^{-k}$. Hence, $k \ge \log(1/\delta)$. Let us get back to the failure probability, and notice that it is bounded by
\[
2e^{-C_0  \log (1/\delta)/C_f}.
\]
Recall that $C_0$ is a constant that we can define, and we can set it sufficiently large such that this failure probability is at most $\delta/2$. This concludes that the bound on the sum of second derivatives is bounded by $\epsilon/2$ with probability $1-\delta/2$. In particular, \eqref{eq:final-bound} holds with probability $1-\delta$ which concludes the proof. 
\end{proof}

\subsection{Proof of Proposition~\ref{prop:heavytail}}\label{subsec:probabilistic}
We restate the following proposition and prove it:

\propConcentrateMartingale*
For convenience, we will refer to a random variable $X$ as $(I(t),C)$ bounded if $\Pr[|X| > t] \le C\exp(-I(t))$ for all $t \ge 0$.

We will use a truncation of the random variables.
Given $L>0$ define $X_i^{\le L} = X_i \mathds{1}(|X_i| \le L)$ and $X_i^{>L} = X_i \mathds{1}(|X_i| > L)$, and notice that $X_i = X_i^{\le L} + X_i^{>L}$. 
We bound the sum $\sum_i X_i$ by considering the moment generating function of $\sum_i X_i^{\le L}$ and bounding the probability that there exists $i$ such that $X_i > L$, as formalized in the following lemma:
\begin{lemma}\label{lem:mgf-and-truncation}
Fix $L,K,\theta,\delta>0$ be such that for all $i$,
\[
\E\lp[e^{\theta X_{i}^{\le L}} \mid X_1 \cdots X_{i-1}\rp] \le e^K ; \quad \text{and }
\Pr[|X_{i}| > L \mid X_1 \cdots X_{i-1}] \le \delta.
\]
Then, for any $t > 0$,
\[
\Pr\lp[\lp|\sum_i X_i\rp| > t\rp]
\le 2 e^{Kn - \theta t} + \delta n.
\]
\end{lemma}
\begin{proof}
First, by a standard induction on $n$, one can show that 
\[
\E\lp[\exp\lp(\theta \sum_{i=1}^n X_{i}^{\le L}\rp)\rp]
\le e^{Kn}.
\]
Consequently, by Markov's inequality,
\[
\Pr\lp[\sum_i X_i^{\le L}> t\rp]
= \Pr\lp[e^{\theta \sum_i X_i^{\le L}}> e^{\theta t}\rp]
\le \E[e^{\theta \sum_i X_i^{\le L}}]/e^{\theta t}
\le e^{Kn - \theta t}.
\]
Similarly, the probability that the sum is less than $-t$ can be bounded by the same quantity. Thus,
\[
\Pr\lp[\lp|\sum_i X_i\rp| > t\rp]
= \Pr\lp[\lp|\sum_i X_i^{\le L} + X_i^{>L}\rp| > t\rp]
\le \Pr\lp[\lp|\sum_i X_i^{\le L}\rp| > t\rp] + \Pr\lp[\exists i,~ X_i^{>L} > 0\rp] \le 2 e^{Kn - \theta t} + n \delta.
\]

\end{proof}

Therefore, we would like to bound the moment generating function of $X_{i+1}$ given $X_1 \cdots X_i$. We have the following lemma:
\begin{lemma}\label{lem:x2-moment-bound}
Let $X$ be a zero-mean random variable. Then, for any $\theta > 0$,
\[
\E[e^{\theta X}] \le 1 + \theta^2 \E[X^2 e^{\theta |X|}]/2.
\]
\end{lemma}
\begin{proof}
For any $z\in \mathbb{R}$, we have by the Taylors series of $e^z$ aronud $z=0$,
\[
e^{z} = 1 + z + \frac{z^2}{2} e^{\zeta(z)}
\le 1 + z + \frac{z^2}{2} e^{|z|}
\]
where $\zeta(z)$ is in the segment between $0$ and $z$. Substituting $z = \theta X$, taking expectation over $X$, and using the fact that $X$ is zero mean, the result follows.
\end{proof}

We would like to apply the above lemma for bounding the moment generating function.
\begin{lemma}\label{lem:bnd-mgf-formula}
Let $X$ be a zero-mean random variable that is $(I(t),C)$ bounded. Then,
\[
\E\lp[X^2 e^{\theta |X|}\rp] \le \int_{0}^{\infty} (2t+\theta t^2) e^{\theta t} \Pr[|X| > t] dt.
\]
\end{lemma}
\begin{proof}
Define $Z = |X|$. Then, our goal is to bound $\E[Z^2 e^{\theta Z}]$. By a standard change of measure argument (that can be proved, e.g., using integration by parts), for any differentiable function $h \colon [0,\infty)$ and any nonnegative r.v. $Z$, one has
\[
\E[h(Z)] = h(0) + \int_0^\infty \Pr[Z>t] \frac{dh(t)}{dt} dt.
\]
Applying for $h(t) = t^2 e^{\theta t}$, and substituting $dh(t)/dt=(2t + \theta t^2) e^{\theta t}$, the result follows. 
\end{proof}

We would like use Lemma~\ref{lem:bnd-mgf-formula} on $X^{\le L}$:
\begin{lemma}\label{lem:bnd-mgf-by-M}
Let $X$ be a zero-mean random variable that is $(I(t),C)$ bounded. Let $L,\theta > 0$ such that $\theta \le I(L)/2L$. Assume that $I(t)/t$ is monotonic decreasing and that $I(t) \le t$. Then, for the value $M$ defined in Proposition~\ref{prop:heavytail},
\[
\E[e^{\theta X^{\le L}}]
\le e^{\theta^2 M/2}.
\]
\end{lemma}
\begin{proof}
Notice that from the monotonicity of $I(t)/t$ and from the fact that $\theta L \le I(L)/2$, we have that for any $0 \le t \le L$,
\[
\theta t = \theta \frac{t}{I(t)}I(t)
\le \theta \frac{L}{I(L)} I(t) \le I(t)/2. 
\]
Using this inequality and the fact that $\theta \le I(L)/2L \le 1/2\le 1$, we obtain
\begin{align*}
\int_{0}^{\infty} (2t+\theta t^2) e^{\theta t} \Pr[|X^{\le L}| > t] dt
\le C\int_{0}^L (2t + t^2) e^{\theta t - I(t)} dt \le 
C\int_{0}^{L} (2t + t^2) e^{-I(t)/2} dt \\
\le C\int_{0}^\infty (2t + t^2) e^{-I(t)/2} dt
\le M.
\end{align*}
Using Lemma~\ref{lem:x2-moment-bound} and Lemma~\ref{lem:bnd-mgf-formula}, it follows that
\[
\E[e^{\theta X^{\le L}}]
\le 1 + \theta^2 M/2 \le e^{\theta^2 M/2}.
\]
\end{proof}

\begin{proof}[Proof of Proposition~\ref{prop:heavytail}]
We apply Lemma~\ref{lem:mgf-and-truncation}, substituting $L,\theta,\delta$, such that $L$ is to be chosen later, $\delta = e^{-I(L)}$ and $\theta$ is a value to be chosen later that satisfies $\theta \le \min(I(L)/2L,1)$. From Lemma~\ref{lem:bnd-mgf-by-M} we can further substitute $K=\theta^2M/2$. We obtain that
\begin{equation}\label{eq:bnd-sum}
\Pr\lp[\lp|\sum_i X_i \rp|>t\rp]
\le 2 e^{\theta^2 Mn/2-\theta t} + Cne^{-I(L)}.
\end{equation}
Let us now substitute $L$ and $\theta$. Let $t^*$ be the solution of $t = MnI(t)/2t$, and notice that there is a unique solution, since the left hand side ($t$) is increasing while the right hand side is decreasing, by the assumption that $I(t)/t$ is decreasing. If $t \le t^*$, we take $L=t^*$ and $\theta = t/Mn$. Notice that
\[
\theta = t/Mn \le t^*/Mn = I(t^*)/2t^* = I(L)/2L,
\]
as required by Lemma~\ref{lem:bnd-mgf-by-M}. Then,
\[
e^{\theta^2 Mn/2-\theta t}
= e^{-t^2 /2Mn},
\]
and substituting into \eqref{eq:bnd-sum} concludes the case $t \le t^*$. 
If $t > t^*$, we substitute $L=t$ and $\theta = I(t)/2t=I(L)/2L$. Then, using the definition of $t^*$, we can bound the right hand side of \eqref{eq:bnd-sum} by
\begin{align*}
2e^{\theta^*Mn/2-\theta t} + Cne^{-I(L)}
=
2e^{I(t)^2Mn/8t^2 - I(t)/2} + Cne^{-I(t)}
= 2e^{I(t)Mn/2t \cdot I(t)/4t - I(t)/2} + Cne^{-I(t)}\\
\le 2e^{I(t)/4-I(t)/2} + Cne^{-I(t)}
= 2e^{-I(t)/4} + Cne^{-I(t)}.
\end{align*}
This concludes the proof.
\end{proof}

\section{Applying the abstract bound: Proof of corollary~\ref{cor:to-abstract}} \label{sec:cor-to-abst}
For the function $f(\eta) = 1/(1-\eta^2)^p$, one can compute that
\[
\lp| \frac{\partial f(\eta)}{\partial \eta} \rp|
= \lp| \frac{2\eta p}{(1-\eta^2)^{p+1}} \rp|
\le \frac{2p}{(1-\eta^2)^{p+1}} \enspace.
\]
Using the function $I(t) = (t/2p)^{p/(p+1)}$, we have that $I(|f'(\eta)|) \le f(\eta)$. It is straightforward to verify that all the other conditions on $I(t)$ and $f$ follow as well. Next, we find $t^*$, which is the solution to $t = kI(t)/2t$, namely to $t = k (t/2p)^{p/(p+1)}/2t$, which is solved by $t^* = C(p) k^{(p+1)/(p+2)}$,
where $C(p)$ depends only on $p$. Lastly, we have
\[
I(t^*) = C'(p) \cdot k^{p/(p+2)},
\]
for some $C'(p)$,
and the guarantee of the theorem implies an optimal rate for any $t \ge \exp(-I(t^*))$.

Next, we study the function $f(\eta) = \exp\lp(\exp\lp(1/(1-\eta^2)\rp)\rp)$. Denote by $h(\eta) = \frac{1}{1-\eta^2}$, and notice that $f(\eta) = \exp(\exp(h(\eta)))$. By the chain rule,
\begin{align*}
	\lp| f'(\eta)\rp| = \lp| \frac{d}{d\eta} \exp(\exp(h(\eta)))\rp|
	= \lp|\exp(\exp(h(\eta))) \exp(h(\eta)) h'(\eta)\rp|
	= \lp| f(\eta) \log(f(\eta)) \frac{2\eta}{(1-\eta^2)^2}\rp|\\
	\le \lp|2 f(\eta) \log(f(\eta)) \cdot h(\eta)^2 \rp|
	= 2 f(\eta) \log(f(\eta)) (\log\log f(\eta))^2.
\end{align*}
Now, the intuition is to take $I(t)$ to be the inverse of $g(u) = 2u \log u \log\log^2 u$, and this will imply that $I(|f'(\eta)|) \le f(\eta)$. This substitution, however, does not satisfy all the required assumptions on $I(t)$. To be more formal, notice that $g(u)$ is monotonic increasing in $[e,\infty)$, that $g(e) = 0$ and that $\lim_{u\to \infty} g(u) = \infty$. Hence, $g \colon [e,\infty) \to [0,\infty)$ has an inverse that we denote by $h\colon [0,\infty) \to [e,\infty)$, which is also monotonic increasing. Further, we argue that $h$ is concave. Letting $g'$, $g''$, $h'$ and $h''$ denote derivatives, one has that
\[
h''(t) = \frac{-1}{(g'(h^{-1}(t))^3 g''(h^{-1}(t))} < 0,
\]
as $g$ is increasing and convex in $[e,\infty)$. Since $g(u) = \omega(u)$ as $u\to \infty$, we have that $h(t) = o(t)$ as $t \to \infty$. Let $t' = \sup_t h(t) \ge t$, and define
\[
I(t) = \begin{cases}
	t & t \le t' \\
	h(t) & t > t'
	\end{cases}.
\]
Then, $I(t) \le t$ as required, and further, $I(t) \ge \Omega(\sqrt{t})$ as $t\to \infty$, since $I(t) = \Theta(h(t)) = \Theta(t/(\log t \log^2\log t))$ as $t\to \infty$. It remains to argue that $I(t)/t$ is decreasing. First, $I(t)$ is a concave function as a minimum of two concave functions, and it satisfies $I(0) = 0$. 
Then, computing the derivative of $I(t)/t$, one has
\[
\frac{d}{dt} \frac{I(t)}{t} 
= \frac{I'(t)t - I(t)}{t^2}
= \int_{0}^t \frac{I'(t) - I'(s)}{t^2} ds
\le 0,
\]
which follows from the fact that $I$ is concave, hence its derivative is decreasing in $t$. This concludes that $I(t)/t$ is decreasing. It is straightforward to verify the other assumptions on $f$ and $I(t)$.

Recall that $I(t) = \Theta(t/(\log t\log^2\log t))$ as $t\to \infty$.
Next, we solve for $t^*$ that is the solution of $t = k I(t)/2t = k/(2\log t \log^2\log t)$. We obtain that $t^* = \Theta(k/(\log k \log^2 \log k))$, and $I(t^*) = \Theta(k/(\log^2 k \log^4 \log k))$, as required.


\section{Lower bound: Proof of Theorem~\ref{thm:lb-informal}}\label{sec:lb}
Below, we state a formal version of Theorem~\ref{thm:lb-informal} and provide its proof.

\begin{theorem}\label{thm:lb-formal}
Fix $k,\delta, \epsilon=1$, $\Delta=1$  and $M>0$, and let $\mu$ be a continuous noise distribution supported on $[-M,M]$ whose density is monotonically decreasing for $\eta\ge 0$ and increasing for $\eta \le 0$.
Assume that the algorithm that adds to each answer an i.i.d. noise drawn from $\mu$, is $(1,\delta)$ differentially private against $1$-sensitive queries. Then, $M \ge \Omega(\log 1/\delta \log k)$. In particular, for $\delta = e^{-k}$, $M \ge \Omega(k \log k)$, and for any $\delta \le e^{-\omega(k/\log^2 k)}$, $M \ge \omega(\sqrt{k\log 1/\delta})$, where $\omega()$ denotes a strict asymptotic inequality.
\end{theorem}


We can assume that $\mu[0,\infty) \ge 1/2$ (otherwise we can consider $-\mu$ instead of $\mu$). Further, recall that $\mu$ is a continuous distribution, hence it has density that we can denote by $e^{-f(x)}$, where $f(x)=\infty$ if the density is zero. Our assumption implies that $f(\eta)$ is increasing for $\eta \ge 0$ and decreasing for $\eta \le 0$.

The argument consists of the following lemmas:

\begin{lemma}\label{lem:base}
Assume that $\Pr_\mu[\eta \ge 0] \ge 1/2$ and that $\delta \le 0.1$. Then,
\[
f(1/2) \le \log(10 M).
\]
\end{lemma}

For an intuitive explanation about this lemma, notice that it is stating that the density at $1/2$ cannot be very small. This follows from two facts: (1) the density at $0$ is at least $1/2M$, since $\mu([0,M]) \ge 1/2$ and the density is decreasing in $[0,M]$. Further, since the noise satisfies $(1,0.1)$-DP, its density cannot drop ``too fast'', other a change in the true query value will be detected. The formal proof of this lemma appears in Section~\ref{pr:lem-base}.

\begin{lemma}\label{lem:step}
Let $\eta\ge 1/2$ such that
\[
\max\lp(\log 2,\frac{\log(1/2\delta)}{4(k-1)}\rp) \le f(\eta) \le \frac{\log (1/2\delta)}{3}.
\]
Then, 
\[
f(\eta+1/2) \le f(\eta)\exp(8/\log(2/\delta)).
\]
\end{lemma}
To gain some intuition on this lemma, we again use the fact that since $\mu$ satisfies DP, the noise density cannot drop too fast. In particular, if the density at $\eta$ is non-negligible, then the density at $\eta+1/2$ cannot drop too fast. Notice that the assumption that the density at $\eta$ is non-negiligible corresponds to requiring that $f(\eta)\le \log(1/\delta)/3$. On the other hand, the lower bound requirement on $f$ weakens as $k$ grows. This is due to the fact that we utilize the multiple samples. The proof appears in Section~\ref{sec:pr-lem-step}.

The proof concludes by combining these two lemmas. Since $f(1/2)$ is relatively and due to the bound on the growth rate of $f$, we conclude that $f(\eta)\le O(\log(1/\delta))$ for some $\eta \ge \Omega(\log k \log 1/\delta)$. In particular, this implies that $[0,\eta]$ is contained in the support of $\mu$, and concludes the proof. The analysis is based on a case analysis as formalized below:

\begin{proof}[Proof of Theorem~\ref{thm:lb-formal}]
To complete the proof, let $\eta_0$ be the minimal $\eta$ such that $\eta \ge 1/2$ and 
\[
f(\eta) \ge \max\lp(\log 2,\frac{\log(2/\delta)}{4(k-1)}\rp).
\]
Then, by Lemma~\ref{lem:base},
\begin{equation}\label{eq:feta0}
f(\eta_0) \le \max\lp(\log 2,\frac{\log(2/\delta)}{4(k-1)}, \log(10M)\rp).
\end{equation}
Applying Lemma~\ref{lem:step} multiple times, we derive that
\[
f(\eta_0 + i/2) \le f(\eta_0) e^{i \cdot 8/\log(2/\delta)},
\]
for any $i$ such that
\[
f(\eta_0) e^{i \cdot 8/\log(2/\delta)}
\le \frac{\log(2/\delta)}{3}.
\]
Equivalently, this holds for any
\[
i \le \ell := \frac{\log( \log(2/\delta)/3 f(\eta_0))}{8/\log(2/\delta)} = \log\lp( \frac{\log(2/\delta)}{3 f(\eta_0)}\rp) \cdot \frac{\log(2/\delta)}{8}\enspace.
\]
In particular, $f(\eta_0 + \lfloor\ell\rfloor/2) < \infty$, which implies that $M > \eta_0 + \lfloor\ell\rfloor/2$. It suffices to show that $\ell \ge \Omega(\log 1/\delta \cdot \log k)$ to conclude the proof. We divide into cases according to $f(\eta_0)$, using \eqref{eq:feta0}.
\begin{itemize}
    \item If $f(\eta_0)\le \log 2$: then, $\ell \ge \Omega(\log(2/\delta)\log\log(2/\delta))$. We divide into cases according to $\delta$: if $\delta \le e^{-k/\log^2 k}$ then $\log\log(2/\delta) \ge \Omega(\log k)$ and the proof follows. Otherwise, we use the theorem of \cite{steinke2016between} that claims that for any $\delta \ge e^{-k}$ and for any $(1,\delta)$ mechanism for $1$-sensitive queries, it holds that the average error is bounded as follows: 
    \[
    \frac{1}{k}\sum_{i=1}^k|q_i(\vx) - a_i| \ge \Omega(\sqrt{k\log(1/\delta)})\enspace.
    \]
    This implies that if the mechanism uses independent bounded noise of magnitude bounded by $M$, then $M \ge \Omega(\sqrt{k\log(1/\delta)})$. Since $\sqrt{k\log(1/\delta)} \ge \log 1/\delta\log k$ for $\delta \le e^{-k/\log^2 k}$, the result follows.
    \item If $f(\eta_0) \le \log(10M)$. As argued for the previous case, we can assume that $\delta \le e^{-k/\log^2 k}$. Further, we can assume that $M \le \log(1/\delta)\log(k)/10$, otherwise the theorem follows. The above two assumptions imply that
    \[
    f(\eta_0)
    \le \log (10M)
    \le \log \log (1/\delta) + \log \log k
    \le O(\log\log(1/\delta)).
    \]
    This implies that 
    \[
    \ell = \log\lp( \frac{\log(2/\delta)}{3 f(\eta_0)}\rp) \cdot \frac{\log(2/\delta)}{8}
    \ge \Omega\lp(\log\log(1/\delta)\log(1/\delta)\rp)
    \ge \Omega(\log k \log(1/\delta)),
    \]
    where the last inequality follows from $\delta \le e^{-k/\log^2 k}$.
    \item If $f(\eta_0) \le \log(2/\delta)/4(k-1)$, it clearly follows from definition of $\ell$ that $\ell \ge \Omega(\log(1/\delta)\log k)$.
\end{itemize}
This concludes the proof.
\end{proof}

\subsection{Proof of Lemma~\ref{lem:base}}\label{pr:lem-base}

We start with a simple lemma that argues that if we shift any subset $U$ of $\mathbb{R}^k$ by any $v \in [0,\Delta]^k$ then the probability of the shifted set should not significantly differ from that of $U$, if the noise satisfies DP.
\begin{lemma} \label{lem:Uplusv}
Let $\epsilon,\delta$, let $\mu^k$ be a noise that is $(\epsilon,\delta)$ differentially private against $1$-sensitive queries. Let $U \subseteq \mathbb{R}^k$, $v \in [0,1]^k$, and define $U+v = \{u+v \colon u \in U\}$. Then.
\[
\mu^k(U) \le e^\epsilon\mu^k(U + v) + \delta.
\]
\end{lemma}
\begin{proof}
Assume that $\calX = [-1,0]$, that $n=1$, and define the query $q_i(x) = x$ for all $i = 1,\dots,k$. Let $x_1 =(0,\dots,0)$ and $x_1' = -v$. Recall that $a_1,\dots,a_k$ are the answers of the algorithm, and notice that by the $(\epsilon,\delta)$ differential privacy,
\[
\Pr[(a_1,\dots,a_k)\in U \mid x_1] \le e^\epsilon \Pr[(a_1,\dots,a_k)\in U \mid x_1'] + \delta.
\]
This is equivalent to
\[
\Pr_\mu[\eta \in U] \le e^\epsilon \Pr_\mu[\eta \in U+v] + \delta,
\]
as required.
\end{proof}
We can conclude with the proof.
\begin{proof}[Proof of Lemma~\ref{lem:base}]
Assume towards contradiction that $f(1/2) > \log(10 M)$, and this implies by monotonicity that $f(1) > \log(10M)$. Then,
\[
\mu[1,\infty)
= \mu[1,M]
= \int_1^{M} e^{-f(\eta)} d\eta
\le \int_1^{M} e^{-f(1)} d\eta
\le M e^{-f(1)}
< 1/10.
\]
Consequently,
\[
\mu[0,1] = \mu[0,\infty) - \mu[1,\infty) > 0.4
\]
while 
\[
\mu[1,2] \le \mu[1,\infty) < 0.1.
\]
Let $U = \{(x_1,\dots,x_n) \colon x_1 \le 1\}$. 
It holds that
\[
\mu^n(U) = \mu[0,1] > 0.4, 
\]
while 
\[
\mu^n(U+(1,0,\dots,0)) = \mu[1,2] < 0.1.
\]
This implies that
\[
\mu^n(U) > e^1 \mu^n(U+(1,0,\dots,0)) + 0.1,
\]
which contradicts Lemma~\ref{lem:Uplusv} and the fact that the noise is $(1,0.01)$ private.
\end{proof}

\subsection{Proof of Lemma~\ref{lem:step}}\label{sec:pr-lem-step}
First we use the following result, which is analogous to a bound that appears in the upper bound in this paper, and follows from Lemma~\ref{lem:Uplusv} above.

\begin{lemma}\label{lem:cannot-be-difference}
If the noise if $\mu^k$ is $(1,\delta)$ private with respect to $1$-sensitive queries, then for any $v_1,\dots,v_k \in [0,1]$,
\begin{equation}\label{eq:lb-sum}
\Pr_{x\sim \mu}\lp[\sum_{i=1}^k f(x_i+v_i) - f(x_i) \ge 2\rp] < 2\delta.
\end{equation}
\end{lemma}
\begin{proof}
Look at the set $U = \{\eta \colon \sum_i f(\eta_i+v_i) - f(\eta_i) \ge 2\}$.
From Lemma~\ref{lem:Uplusv} and the $(1,\delta)$ privacy assumption, we have
\[
\Pr_{\mu^k}[U]\le e\Pr_{\mu^k}[U+v]+\delta.
\]
Further, 
\[
\Pr_{\mu^k}[U]
= \int_{U} \exp\lp( 
- \sum_{i=1}^k f(\eta_i) d\eta
\rp)
\ge \int_{U} \exp\lp(-2
- \sum_{i=1}^k f(\eta_i+v_i) d\eta
\rp)
= e^2 \Pr_{\eta}[U+v].
\]
Combining the above inequalities, we derive that
\[
\Pr_\mu[U] \le e\Pr_\mu[U+v] + \delta \le \Pr_\mu[U]/e + \delta,
\]
hence
\[
\Pr_\eta[U] (1-1/e) \le \delta,
\]
which implies that $\Pr_\eta[U] < 2\delta$ as required.
\end{proof}

Recall that we want to bound $f(\eta+1/2)-f(\eta)$. In the proof, we assume towards contradiction that $f(\eta_0+1/2)-f(\eta_0)$ is large for some appropriate $\eta_0$ and we will derive that \eqref{eq:lb-sum} fails to hold, which, by Lemma~\ref{lem:cannot-be-difference} implies that the mechanism is not DP. As a first step, we will prove that if $f(\eta_0+1/2)-f(\eta_0)$ is large then a variant of \eqref{eq:lb-sum} is not satisfied, with different constants and $k=1$.


\begin{claim}\label{cla:simplestep}
Let $\eta_0 \ge 1/2$. Then,
\[
\Pr_{\eta\sim \mu}[f(\eta+1)-f(\eta) \ge f(\eta_0+1/2)-f(\eta_0)] \ge \frac{e^{-f(\eta_0)}}{2}\enspace.
\]
\end{claim}
\begin{proof}
We have by monotonicity of $f(\eta)$,
\[
\Pr_{\eta\sim\mu}[\eta_0 - 1/2 \le \eta \le \eta_0]
= \int_{\eta_0-1/2}^{\eta_0} e^{-f(\eta)} d\eta
\ge \int_{\eta_0-1/2}^{\eta_0} e^{-f(\eta_0)} d\eta
= \frac{e^{-f(\eta_0)}}{2}\enspace.
\]
For any such $\eta \in [\eta_0-1/2,\eta_0]$ we have
\[
f(\eta) \le f(\eta_0) \le f(\eta_0+1/2) \le f(\eta+1).
\]
Consequently, $f(\eta+1)-f(\eta) \ge f(\eta_0+1/2) - f(\eta_0)$.
\end{proof}

Next, we extend Claim~\ref{cla:simplestep} to show that if $f(\eta_0+1/2)-f(\eta_0)$ is large for some $\eta_0$, then a variant of \eqref{eq:lb-sum} does not hold, where the sum is over $m>1$ elements. To achieve that, we first use Claim~\ref{cla:simplestep} to show that $\Pr[f(\eta+1)-f(\eta) \ge a] \ge b$ for some $a,b>0$ and then derive that if $\eta_1,\dots,\eta_m$ are i.i.d., then 
\[
\Pr[\sum_{i=1}^m f(\eta_i+1)-f(\eta_i) \ge ma]
\ge \Pr[\forall i\le m, f(\eta_i+1)-f(\eta_i) \ge a]
= \prod_{i=1}^m \Pr[f(\eta_i+1)-f(\eta_i)\ge a]
= b^m.
\]
Choosing $a$ and $b$ appropriately yields the desired result.
\begin{lemma} \label{lem:tensorize}
Let $\eta_0 \ge 1/2$, let $\delta_0 > 0$, and let $C > 0$. Assume that
\[
f(\eta_0+1/2) \ge f(\eta_0)(1+4C/\log(1/\delta_0))
\]
and that 
\[
\max\lp(\log 2,\frac{\log(1/\delta_0)}{4(k-1)}\rp) \le f(\eta_0) \le \frac{\log (1/\delta_0)}{3}.
\]
Then, there is $m \le k$ such that
\[
\Pr_{\veta\sim \mu^k}\lp[\sum_{i=1}^m f(\eta_i+1)-f(\eta_i) \ge C\rp] \ge \delta_0.
\]
\end{lemma}
\begin{proof}
Define $K = 4Cf(\eta_0)/\log(1/\delta_0)$ and $L = f(\eta_0) + \log 2$. Applying Claim~\ref{cla:simplestep}, we have
\[
\Pr_{\eta\sim \mu}[f(\eta+1)-f(\eta) \ge K] 
\ge \Pr_{\eta\sim \mu}[f(\eta+1)-f(\eta) \ge f(\eta_0+1/2)-f(\eta_0)] \ge e^{-L}.
\]
Let $m = \lceil C/K\rceil$. 
First, notice that $m \le k$: indeed, it suffices to show that $C/K+1 \le k$, which holds since
\[
f(\eta_0) \ge \frac{\log 1/\delta_0}{4 (k-1)}\enspace.
\]
Then,
\[
\Pr_{\veta \sim \mu^k}[\sum_{i=1}^m f(\eta_i+1)-f(\eta_i) \ge C] \ge
\Pr_{\veta \sim \mu^k}[\forall i\in\{1,\dots,m\},~ f(\eta_i+1)-f(\eta_i) \ge K]
\ge e^{-Lm}.
\]
It remains to argue that $e^{-Lm} \ge \delta_0$, or equivalently, $Lm \le \log(1/\delta_0)$. By definition of $m$, it suffices to show that $L(C/K+1) \le \log 1/\delta_0$. Indeed, using the fact that $f(\eta_0) \ge \log 2 \ge 1/2$ and that $f(\eta_0) \le \log(1/\delta_0)/3$,
\[
L(C/K+1) = LC/K + L =\log 1/\delta_0 \frac{f(\eta_0)+\log 2}{4 f(\eta_0)} + f(\eta_0) + \log 2 
\le \log 1/\delta_0 \frac{2 f(\eta_0)}{4f(\eta_0)} + 2 f(\eta_0) \le  3f(\eta_0) \le \log(1/\delta_0)
\]
This concludes the proof.
\end{proof}

\begin{proof}[Proof of Lemma~\ref{lem:step}]
Assume towards contradiction the existence of such $\eta$.
Then, since $e^x \ge 1+x$ for all $x$, we have
\[
f(\eta+1/2) \ge f(\eta)(1+8/\log(2/\delta)).
\]
Applying Lemma~\ref{lem:tensorize} with $C=2$ and $\delta_0=2\delta$, we obtain that there eixsists $m \le k$ such that
\[
\Pr_{\veta\sim \mu^k}\lp[\sum_{i=1}^m f(\eta_i+1)-f(\eta_i) \ge 2\rp] \ge 2\delta.
\]
However, by Lemma~\ref{lem:cannot-be-difference} this does not hold. We derive the contradiction, and this concludes the proof.
\end{proof}

\section{Adaptive data analysis: Proof of Corollary~\ref{cor:ada}}\label{sec:ada-proof}

We use the following transfer theorem from \cite{jung2020new, bassily2021algorithmic}:
\begin{theorem}\label{thm:transfer}
Assume that $\mathcal{A}$ is an algorithm that answers $k$ statistical queries, $q_i \colon X \to [0,1]$ given some dataset $(x_1,\dots,x_n)\in X^n$. Further, assume that the algorithm is $(\epsilon,\delta)$-differentially private with respect to its dataset and that with probability $1-\beta'$, all of its answers $a_i$ are $\alpha'$-accurate with respect to the sample, namely,
\[
\Pr\lp[ \forall i=1,\dots,k\colon \lp| a_i - \frac{1}{n}\sum_{j=1}^n q_i(x_j)\rp|\le \alpha' \rp] \ge 1-\beta' \enspace.
\]
Assume that $x_1,\dots,x_n$ are drawn from some distribution $P$. Then, for any $c,d>0$,
the algorithm $\A$ produces answers that are $\alpha(c,d)$-accurate with respect to $P$ with probability $1-\beta(c,d)$, where
\[
\alpha(c,d) = \alpha' + e^\epsilon - 1 + c + 2d; \quad
\beta(c,d) = \beta'/c + \delta/d.
\]
Namely,
\[
\Pr\lp[ \forall i=1,\dots,k\colon \lp| a_i - q_i(P)\rp|\le \alpha(c,d) \rp] \ge 1-\beta(c,d) \enspace.
\]
\end{theorem}
From this, we can easily derive our theorem:
\begin{proof}[Proof of Corollary~\ref{cor:ada}]
Fix $\alpha, \beta \in (0,1/2)$.
We apply Theorem~\ref{thm:transfer}, using the bounded noise mechanism from Theorem~\ref{thm:main-informal}, that answers each query $i$ with $a_i=\sum_{i=j}^n q_i(x_j)+\eta_i$ where $\eta_i$ are i.i.d. bounded noise. We set the privacy parameters to $\epsilon = \alpha/8$ and $\delta=\alpha\beta/4$. We obtain that the answers are $\alpha'$-accurate with respect to the sample, for
\[
\alpha' = O\lp(\frac{\sqrt{k\log(1/\alpha\beta)}}{\epsilon n}\rp),
\]
using Theorem~\ref{thm:main-informal} and the fact that the statistical queries on a dataset of size $n$ are $\Delta=1/n$-sensitive. This holds with probability $1$, hence, we can substitute $\beta'=0$. We take $c\to 0$ and $d=\alpha/4$, and we derive that
\[\lim_{c\to 0}\alpha(c,d) = \alpha' + e^{\alpha/8}-1 + \alpha/2 \le \alpha' + 3\alpha/4; \quad
\forall c>0,\ \beta(c,d) = \beta. \]
Here, we used $e^x \le 1+2x$ for $x\in [0,1]$.
If we take 
\[
n = \Theta\lp(\frac{\sqrt{k \log(1/\delta)}}{\epsilon^2}\rp)
= \Theta\lp( \frac{\sqrt{k \log(1/\alpha\beta)}}{\alpha^2}\rp),
\]
then we have $\alpha' < \alpha/4$, hence, $\alpha(c,d) < \alpha$ for some $c>0$ and $\beta(c,d)=\beta$. We derive by Theorem~\ref{thm:transfer} that the protocol is $\alpha$-accurate with respect to $P$ with probability $1-\beta$, as required.
\end{proof}
\section{Computing tighter upper bounds}\label{sec:exp-formal}

Here we explain how to derive an algorithm that upper bounds the optimal noise level, for each given $\epsilon,k$ and $\delta$. The final algorithm is given as Algorithm~\ref{alg:FindNoise} below, yet, we start by explaining it step by step. First of all, we note the following sufficient condition for $(\epsilon,\delta)$-privacy, which is a tighter variant of its analogoue in the proof of Theorem~\ref{thm:general-func}:

\begin{lemma}\label{lem:simulation-general}
Let $P$ and $Q$ be probability distributions over $\mathbb{R}^k$ with densities $p(x)$ and $q(x)$, respectively. Let $\epsilon, \delta > 0$, and assume that
\[
\int_{\epsilon}^\infty \Pr_{X\sim P}\lp[\log \frac{p(X)}{q(X)} > t \rp] e^{\epsilon-t} dt \le \delta.
\]
Then, for any $U \subseteq \mathbb{R}^d$,
\[
\Pr_{X\sim P}[X \in U] \le e^\epsilon \Pr_{X\sim Q}[X \in U] + \delta.
\]
\end{lemma}
\begin{proof}
First of all, notice that by change of variables $s=t-\epsilon$, one has
\[
\int_{0}^\infty \Pr_{X\sim P}\lp[\log \frac{p(X)}{q(X)}-\epsilon > s \rp] e^{-s} ds \le \delta \enspace.
\]
Assume that the above statement holds and denote by $\Lambda$ the random variable 
\[
\Lambda = \max\lp(0, \log \frac{p(X)}{q(X)}-\epsilon\rp)
\]
where $X\sim P$.
The left hand side translates to
\[
\int_{0}^\infty \Pr[\Lambda > t] e^{-t} dt \le \delta.
\]
We use the known technique of integration by parts for probability distributions, which states that for a nonnegative random variable $Z$ and for a function $F \colon [0,\infty) \to \mathbb{R}$ with a continuous derivative that satisfies $F(0)=0$,
\[
\E[F(Z)] = \int_0^\infty \Pr[Z>t] F'(t) dt.
\]
Substituting $Z=\Lambda$ and $F(t) = 1-e^{-t}$, we derive that
\[
\E[1-e^{-\Lambda}] \le \delta.
\]
Using the fact that $1-e^{-0}=0$ and the fact that $\Lambda \ge 0$, we derive that
\[
\E[(1-e^{-\Lambda}) \mathds{1}(\Lambda > 0)] \le \delta.
\]
Substituting the value of $\Lambda$, we obtain
\[
\int_{\mathbb{R}^k} p(x)  \lp(1- e^\epsilon \frac{q(x)}{p(x)}\rp)  \mathds{1}\lp(\log\frac{p(x)}{q(x)} > \epsilon\rp) dx
= \int_{\mathbb{R}^k} \lp(p(x) - e^\epsilon q(x)\rp) \mathds{1}\lp(p(x) > e^\epsilon q(x)\rp) dx \le \delta.
\]
This implies that for any $U \subseteq \mathbb{R}^k$,
\begin{align*}
\int_U p(x) - e^\epsilon q(x) dx
&\le \int_U \lp(p(x) - e^\epsilon q(x)\rp) \mathds{1} \lp(p(x) > e^\epsilon q(x) \rp) dx\\
&\le \int_{\mathbb{R}^k} \lp(p(x) - e^\epsilon q(x)\rp) \mathds{1} \lp(p(x) > e^\epsilon q(x) \rp) dx
\le \delta,
\end{align*}
as required.
\end{proof}
Let us apply the above lemma for answering multiple queries. Below, we assume for simplicity that the queries are non-interactive, namely, $q_1,\dots,q_k$ are given a-priori. Yet, one can obtain the \emph{exact} same bounds while assuming that they are asked adaptively. We refer to the proof of Theorem~\ref{thm:general-func} (and particularly, to Lemma~\ref{lem:mgf-and-truncation}).
\begin{lemma}
Let $M$ be a mechanism that answers $k$ fixed $\Delta$-sensitive queries by adding an i.i.d. noise drawn from some distribution $\mu$ over $\mathbb{R}$, with density
\[
\mu(\eta) = \frac{e^{-f(\eta)}}{Z}, \quad \text{where } Z = \int_{\mathbb{R}} e^{-f(\eta)} d\eta.
\]
(Here, we use the convention $f(\eta)=\infty$ if $\mu(\eta)=0$.)
Let $\epsilon, \delta > 0$ and assume that for all $v_1,\dots,v_k \in [-\Delta,\Delta]$, it holds that
\begin{equation}\label{eq:exp-bnd-difference}
\int_{\epsilon}^\infty \Pr_{\eta_1,\dots,\eta_k \simiid \mu}\lp[\sum_{i=1}^k 
    f(\eta_i + v_i) - f(\eta_i) > t
\rp]e^{\epsilon-t} dt \le \delta.
\end{equation}
Then, the mechanism is $(\epsilon,\delta)$-private.
\end{lemma}
\begin{proof}
Let $\vx_1$ and $\vx_2$ be two neighboring datasets, let $q_1,\dots, q_k$ denote the queries. Denote the densities of the output of the mechanism on $\vx_1$ and $\vx_2$ by $p_1$ and $p_2$, respectively. Our goal is to show that for all $U \subseteq \mathbb{R}^d$,
\begin{equation*}
\int_U p_1(y_1,\dots,y_k) dy \le e^\epsilon \int_U p_2(y_1,\dots,y_k)dy + \delta.
\end{equation*}
In order to show this inequality, from Lemma~\ref{lem:simulation-general} it suffices to show that
\begin{equation}\label{eq:simulation-toshow}
\int_\epsilon^\infty\Pr_{(y_1,\dots,y_k) \sim p_1}\lp[\log\frac{p_1(y_1,\dots,y_k)}{p_2(y_1,\dots,y_k)} > t \rp]e^{\epsilon-t}dt \le \delta.
\end{equation}
Towards this goal, notice that
\[
p_j(y_1,\dots,y_k) = \prod_{i=1}^k \mu(y_i - q_i(\vx_j)).
\]
In particular,
\begin{equation}\label{eq:toshow2}
\log \frac{p_1(y_1,\dots,y_k)}{p_2(y_1,\dots,y_k)}
= \sum_{i=1}^k f(y_i - q_i(\vx_2)) - f(y_i - q_i(\vx_1)).
\end{equation}
Assuming that $(Y_1,\dots,Y_k)$ is a random variable denoting the output of the mechanism on $\vx_1$, we have that $Y_i - q_i(\vx_1)$ are distributed i.i.d. according to $\mu$. Hence, the right hand side of \eqref{eq:toshow2} is distributed according to
\[
\sum_{i=1}^k f(\eta_i + q_i(\vx_1)-q_i(\vx_2)) - f(\eta_i),
\]
where $\eta_1,\dots,\eta_k \simiid \mu$. Denote $v_i = q_i(\vx_1)-q_i(\vx_2)$ and
since the queries are $\Delta$-sensitive, $|v_i|\le \Delta$. We have that the right hand side of \eqref{eq:toshow2} equals
\[
\sum_{i=1}^k f(\eta_i+v_i)-f(\eta_i).
\]
Combining with the assumption of this lemma, this proves \eqref{eq:simulation-toshow}, which concludes the proof.
\end{proof}

Next, we show how to bound the deviations of $\sum_i f(\eta_i+v_i) - f(\eta_i)$. A standard way is to use the moment generating function, however, for the bounded noises suggested in this paper, the corresponding MGF might not exist. Instead, one can use truncation. In particular, we will find some threshold $L$ such that $\Pr[|\eta_i| > L] \le \delta_1/k$, for some $\delta_1<\delta$. This implies that $\Pr[\exists i=1,\dots,k\colon\ |\eta_i| > L] \le \delta_1$. Then, the left hand side of  \eqref{eq:exp-bnd-difference} can be bounded by
\begin{align*}
\int_{\epsilon}^\infty \lp(\Pr_{\eta_1,\dots,\eta_k \simiid \mu}\lp[\sum_{i=1}^k 
    f(\eta_i + v_i) - f(\eta_i) > t ~\wedge~ 
    \max_i |\eta_i| \le L
\rp] + \Pr[\exists i,\ |\eta_i| > \ell]\rp)e^{\epsilon-t} dt \\
\le \int_{\epsilon}^\infty \Pr_{\eta_1,\dots,\eta_k \simiid \mu}\lp[\sum_{i=1}^k 
    f(\eta_i + v_i) - f(\eta_i) > t ~\wedge~ 
    \max_i |\eta_i| \le L
\rp]e^{\epsilon-t} dt + \delta_1,
\end{align*}
using the fact that $\Pr[\exists i,\ |\eta_i| > \ell] \le \delta_1$ and $\int_\epsilon^\infty e^{\epsilon-t} = 1$. Denote $X_i = (f(\eta_i+v_i)-f(\eta_i))\mathds{1}(|\eta_i|\le L)$ and denote $\delta_2 = \delta - \delta_1$. It is sufficient to prove that
\begin{equation}\label{eq:exp-delta2}
\int_\epsilon^\infty \Pr\lp[\sum_i X_i > t \rp] e^{\epsilon-t} dt
\le \delta_2.
\end{equation}
We bound the deviations of $\sum_i X_i$ using the moment generating function technique: for any $\lambda > 0$, by Markov's inequality, one has
\begin{equation}
\label{eq:exp-mgf}
\Pr\lp[\sum_i X_i > t\rp]
= \Pr\lp[e^{\lambda \sum_i X_i} > e^{\lambda t} \rp]
\le \frac{\mathbb{E}\lp[e^{\lambda \sum_i X_i}\rp]}{e^{\lambda t}}
= \frac{\prod_i \mathbb{E}\lp[e^{\lambda X_i}\rp]}{e^{\lambda t}}.
\end{equation}
Recall that $X_i = (f(\eta_i+v_i)-f(\eta_i))\mathds{1}(|\eta_i|\le L)$. We would like to eliminate the dependence on $v_i \in [-\Delta,\Delta]$. We will bound its MGF as follows:
\begin{lemma}
Assume that $\mu$ is distribution with density $\mu(\eta)\propto e^{-f(\eta)}$ defined on $[-R,R]$. Assume that $\mu$ is log concave, namely, $\mu((\eta_1+\eta_2)/2)\ge \sqrt{\mu(\eta_1)\mu(\eta_2)}$. Further assume that $\mu$ is symmetric, namely, $\mu(\eta)=\mu(-\eta)$. Let $\lambda,L > 0$. Define for any $|t| < R-L$ the random variable
\[
X_t = (f(\eta+t)-f(\eta))\mathds{1}(|\eta|\le L),
\]
where $\eta \sim \mu$.
Then, for any $|a| \le |b| < R - L$,
\[
\E[\exp(\lambda X_a)] \le \E[\exp(\lambda X_b)]\enspace
\]
(notice that $X_t$ is undefined for $|t| > R-L$).
\end{lemma}
\begin{proof}
It is sufficient to assume that $f$ has a continuous derivative. Otherwise, one can approximate $f$ with a sequence of functions with continuous derivatives.
We will show that $\E[e^{\lambda X_t}]$ is monotone increasing when $t \ge 0$, and the result will follow since $\E[e^{\lambda X_t}] = \E[e^{\lambda X_{-t}}]$, as $\mu$ is symmetric. Let $E$ be the event that $|\eta|\le L$. Then, from symmetricity of $\mu$ and $f$,
\begin{align*}
\E\lp[e^{\lambda X_t}\rp]
&= \E\lp[e^{\lambda f(\eta+t)} e^{-\lambda f(\eta)}\mathds{1}_E\rp]
= \E\lp[ \frac{e^{\lambda f(\eta+t)}+e^{\lambda f(-\eta+t)}}{2}e^{-\lambda f(\eta)}\mathds{1}_E\rp]\\
&\E\lp[ \frac{e^{\lambda f(|\eta|+t)}+e^{\lambda f(|\eta|-t)}}{2}e^{-\lambda f(\eta)}\mathds{1}_E\rp].
\end{align*}
It is sufficient to show that the following derivative is nonnegative, for any $\eta \ge 0$:
\begin{equation}\label{eq:exp-derivative}
\frac{d}{dt} \frac{e^{\lambda f(\eta+t)}+e^{\lambda f(\eta-t)}}{2}
= \frac{f'(\eta+t) e^{\lambda f(\eta+t)}-f'(\eta-t) e^{\lambda f(\eta-t)}}{2}.
\end{equation}
Since $\mu$ is symmetric and log-concave, then $f$ is convex and symmetric. In particular, this implies that $f$ has a minimum at $0$ and it is monotonic nondecreasing at $\eta > 0$. Since we assumed that $\eta,t\ge0$, this implies that $\eta+t \ge |\eta-t|$, which implies that $f(\eta+t) \ge f(\eta-t)$. Further, convexity of $f$ implies that its derivative is increasing, which implies that
\[
f'(\eta+t) \ge f'(|\eta-t|) = |f'(\eta-t)|,
\]
using the fact that the derivative of a symmetric function is antisymmetric. The above implies that
\[
f'(\eta+t) e^{\lambda f(\eta+t)} 
\ge \lp| f'(\eta-t) e^{\lambda f(\eta-t)} \rp|,
\]
which derives that \eqref{eq:exp-derivative} is nonnegative and concludes the proof.
\end{proof}

Define
\[
X := (f(\eta+\Delta)-f(\eta))\mathds{1}(|\eta|\le L),
\]
where $\eta\sim \mu$, and the above lemma implies that $\E[e^{\lambda X_i}] \le \E[e^{\lambda X}]$ for all $i=1,\dots,k$. Combining with \eqref{eq:exp-mgf}, one has
\[
\Pr\lp[\sum_i X_i > t\rp] \le \inf_{\lambda>0} \frac{\mathbb{E}\lp[e^{\lambda X}\rp]^k}{e^{\lambda t}}
\le \exp\lp(\inf_{\lambda>0} k \log \lp(\mathbb{E}\lp[e^{\lambda X}\rp]\rp) - \lambda t  \rp).
\]
The function $k \log \lp(\mathbb{E}\lp[e^{\lambda X}\rp]\rp) - \lambda t$ is known to be convex in $\lambda$ for any random variable $X$, hence, it can be optimized efficiently, and any $\lambda$ would yield an upper bound. Integrating over $t$, one can bound the left hand side of \eqref{eq:exp-delta2}. 
This produces a way to certify that \eqref{eq:exp-delta2} holds, for given $\epsilon,\delta,k$ and $\mu$. Recall that if this inequality holds, then the mechanism is guaranteed to be $(\epsilon,\delta)$-DP. In particular, given $f$ and $R$, \eqref{eq:exp-delta2} certifies that $M_{f,R}$ is $(\epsilon,\delta)$-DP. 
If we want to find an upper bound on the minimal magnitude $R$ such that $M_{f,R}$ is $(\epsilon,\delta)$-DP, we can perform a simple binary search over values of $R>0$ (stopping when the desired precision has achieved).
We note that in order to obtain a proper upper bound, one has to ensure that the approximation errors in the relevant computations are one sided (e.g. the computed MGF value should not be lower than the actual value). Algorithm~\ref{alg:CheckDP} tests if a mechanism is $(\epsilon,\delta)$-DP and Algorithm~\ref{alg:FindNoise} finds an upper bound on the minimal noise $R$ given some function $f$.

\SetKwProg{Fn}{Function}{:}{\KwRet}
\SetKwInput{kwInput}{Input}
\SetKwInput{kwOutput}{Output}
\SetKwFunction{FTestPriv}{testPrivacy}
\begin{algorithm}[h!]
\SetKwInput{kwTunableParam}{Tunable Parameter}
\SetKwFunction{FMGF}{MGF}
\SetKwFunction{FMain}{main}
\SetKwFunction{FDevBnd}{deviationBound}
\Fn(\tcc*[h]{Checks if a mechanism is $(\epsilon,\delta)$-private}){\FTestPriv{$\epsilon,\delta,k,\Delta,p$}}{
\kwInput{Privacy parameters $\epsilon>0$ and $\delta \in (0,1)$; Number of queries $k \in \mathbb{N}$; sensitivity $\Delta>0$; A probability density $p$ over $(-R,R)$ such that $\log p(y)$ is concave and $p(y)=p(-y)$.}
\kwOutput{An indication whether the mechanisms that answers $k$ $\Delta$-sensitive queries with i.i.d. noise according to $p$ is $(\epsilon,\delta)$-DP. A \textbf{false} answer may be wrong but a \textsf{true} answer is always correct.}
\Fn(\tcc*[h]{Computes a moment generating function}){\FMGF{$L, \lambda$}}{
\kwInput{A threshold $L > 0$ and a real parameter $\lambda > 0$}
\kwOutput{The moment generating function $\mathbb{E}[e^{\lambda X}]$, where $Y\sim p$ and $X = (\log p(Y)-\log p(Y+\Delta))\mathds{1}(|Y|\le L)$}
\Return{$\int_{-L}^L p(y) \exp(\lambda(\log p(y)-\log p(y+\Delta)))dy$} \tcc*[h]{An upper bound is also valid and can be computed since $\log p(y)-\log p(y+\Delta)$ is monotone increasing.}
}
\Fn(\tcc*[h]{Computes a probability of deviation}){\FDevBnd{$L,t$}}{
\kwInput{Real numbers $L,t>0$}
\kwOutput{An upper bound on the probability that $\sum_{i=1}^k X_i > t$, where $Y_1,\dots,Y_n\simiid p$, $v_i \in [-\Delta,\Delta]$ are arbitrary and $X_i = (\log p(Y_i)-\log p(Y_i+v_i))\mathds{1}(|Y_i|\le L)$}
logProb $\gets \inf_{\lambda>0} k \FMGF(L,\lambda)-\lambda t$ \;
\Return $e^{\text{logProb}}$\;
}
\kwTunableParam{$\delta_1 \in (0,\delta)$.} \tcc*[h]{Can be set arbitrarily. A possible setting is: $\delta_1=0.01\delta$}\;
$L \gets$ the unique value in $[0,R]$ such that
$\int_{-L}^L p(y)dy = 1-\delta_1$ \tcc*[h]{An upper bound is also valid.}\;
\If{$L+\Delta \ge R$}{
\Return{\textsf{false}}}
$\delta_2 \gets \int_\epsilon^\infty \FDevBnd(L,t)e^{\epsilon-t}dt$ \tcc*[h]{An upper bound is valid and can be computed since deviationBound is monotone decreasing in $t$}\;
\eIf{$\delta_1+\delta_2\le\delta$}{
\Return{\textsf{true}}
}{
\Return{\textsf{false}}
}
\caption{Check if a mechanism satisfies DP\label{alg:CheckDP}}
}
\end{algorithm}

\begin{algorithm}[h!]
\SetKwInput{kwDef}{Definition}
\SetKwFunction{FNoiseUB}{noiseUpperBound}
\Fn(\tcc*[h]{Computes an upper bound on the noise required for preserving a desired privacy level}){\FNoiseUB{$\epsilon,\delta,k,\Delta,p_1$}}{
\kwInput{Privacy parameters $\epsilon>0$ and $\delta \in (0,1)$; Number of queries $k$; Sensitivity $\Delta$; A probability density function $p_1\colon (-1,1)\to (0,\infty)$ such that $\log p_1$ is concave and $p_1(y)=p_1(-y)$.}
\kwDef{For all $R>0$, define by $p_R$ the density over $[-R,R]$, such that $p_R(y) = p_1(y/R)/R$.}\tcc{Equivalently, a sample $y\sim p_R$ is obtained by sampling $y'\sim p_1$ and outputting $Ry'$.}
\kwOutput{A number $R>0$ such that $p_R$ is $(\epsilon,\delta)$-DP for answering $k$ $\Delta$-sensitive queries.}\;
err$\gets$ a very small number \tcc{For example, $10^{-8}$}
\tcc{Compute an upper bound $b$ on the minimal allowable noise $R$.}
$b \gets 1$ \;
\While{{\bf not} \FTestPriv{$\epsilon,\delta,k,\Delta,p_R$}}{
$b \gets 2 * b$
}
$a \gets 0$ \tcc{A lower bound on the minimal allowable noise $R$}
\While{$b-a>err$}{
$m\gets (a+b)/2$\;
\eIf{\FTestPriv{$\epsilon,\delta,k,\Delta,p_{m}$}}{
$b\gets m$
}{
$a \gets m$
}
}
\Return{$b$}
}
\caption{Find a suitble noise magnitude $R$\label{alg:FindNoise}}
\end{algorithm}

\clearpage
\appendix

\printbibliography

\end{document}